\documentclass[11pt,a4paper]{article}
\usepackage{a4wide}

\usepackage{enumerate}
\usepackage{eucal}
\usepackage{todonotes}

\usepackage[section]{algorithm}
\usepackage[noend]{algpseudocode}

\usepackage[fleqn]{amsmath}
\usepackage{amstext,amssymb,amsthm}
\usepackage{xspace}

\floatname{algorithm}{Pseudocode}

\newcommand{\ceil}[1]{{\lceil {#1} \rceil}}
\newcommand{\floor}[1]{{\lfloor {#1} \rfloor}}

\newcommand{\ignore}[1]{}
\DeclareMathOperator{\la}{la}

\renewcommand{\Pr}[1]{\operatorname{Pr}\mathopen{}\left( #1 \right)\mathclose{}}
\newcommand{\Expect}[1]{\operatorname*{E}\mathopen{}\left( #1 \right)\mathclose{}}

\newtheorem {theorem}{Theorem}[section]

\newtheorem {proposition}[theorem]{Proposition}
\newtheorem {lemma}[theorem]{Lemma}
\newtheorem {corollary}[theorem]{Corollary}

\newcommand{\mycase}[1]{\noindent{\bf Case #1\ }}
\newcommand{\mon}{{\text{mon}}}
\newcommand{\monW}{{\text{monW}}}
\newcommand{\treeroot}{{\text{root}}}
\newcommand{\children}{{\text{children}}}
\newcommand{\depth}{{\text{depth}}}
\newcommand{\id}{{\text{id}}}
\newcommand{\parent}{{\text{parent}}}

\newcommand{\field}[1]{\textup{GF}(#1)}
\newcommand{\x}{{\bf x}}
\newcommand{\y}{{\bf y}}
\newcommand{\z}{{\bf z}}

\newcommand{\ProblemFormat}[1]{{\sc #1}}
\newcommand{\ProblemName}[1]{\ProblemFormat{#1}\xspace}

\newcommand{\probHamilton}{\ProblemName{Hamiltonicity}}
\newcommand{\probHamPath}{\ProblemName{Hamiltonian Path}}
\newcommand{\probkPath}{\ProblemName{$k$-Path}}
\newcommand{\probkIST}{\ProblemName{$k$-IST}}

\newcommand{\heading}[1]{\medskip\noindent{\bf #1.\ }}%

\begin{document}

\title{Spotting Trees with Few Leaves\thanks{Work partially supported by the National Science Centre of Poland, grant number UMO-2013/09/B/ST6/03136 (\L{}K).}}
\author{Andreas Bj\"orklund\thanks{Lund University, Lund, Sweden}, Vikram Kamat\thanks{University of Warsaw, Poland}, {\L}ukasz Kowalik\thanks{University of Warsaw, Poland}, Meirav Zehavi\thanks{Technion IIT, Haifa, Israel}
}

\date{}

\maketitle

\renewcommand{\L}{\mathcal{L}}

\begin{abstract}
We show two results related to the \probHamilton and \probkPath algorithms in undirected graphs by Bj\"orklund [FOCS'10], and Bj\"orklund et al.,  [arXiv'10]. First, we demonstrate that the technique used can be generalized to finding some $k$-vertex tree with $l$ leaves in an $n$-vertex undirected graph in $O^*(1.657^k2^{l/2})$ time. It can be applied as a subroutine to solve the {\sc $k$-Internal Spanning Tree ($k$-IST)} problem in $O^*(\operatorname{min}(3.455^k, 1.946^n))$ time using polynomial space, improving upon previous algorithms for this problem.
In particular, for the first time we break the natural barrier of $O^*(2^n)$.
Second, we show that the iterated random bipartition employed by the algorithm can be improved whenever the host graph admits a vertex
 coloring with few colors; it can be an ordinary proper vertex coloring, a fractional vertex coloring, or a vector coloring.
 In effect, we show improved bounds for \probkPath and \probHamilton in any graph of maximum degree $\Delta=4,\ldots,12$ or with vector chromatic number at most~$8$.
\end{abstract}

\section{Introduction}

In this paper we study the {\sc $(k,l)$-Tree} problem. Given an undirected host graph $G$ on $n$ vertices, {\sc $(k,l)$-Tree} asks if $G$ contains a tree $T$ on $k$ vertices, such that the number of leaves in $T$ is exactly $l$. This problem is a natural generalization of the classic \probkPath problem: for $l=2$, the definitions of {\sc $(k,l)$-Tree} and \probkPath coincide.
For $k=n$ and $l=2$ we get the classic \probHamPath.
Furthermore, {\sc $(k,l)$-Tree} is tightly linked to the well-studied {\sc $k$-Internal Spanning Tree ($k$-IST)} problem, which asks if a given graph $G$ on $n$ vertices contains a spanning tree $T$ with at least $k$ internal vertices. Indeed, it is well-known that a yes-instance of {\sc $k$-IST} is a yes-instance of {\sc $(k+l,l)$-Tree} for some $l\leq k$, and vice versa (see Section~\ref{sec:kIST}).

Because of the connections to \probHamPath, the {\sc $(k,l)$-Tree}, {\sc $k$-IST} and \probkPath problems, even in bipartite graphs or in graphs of bounded degree 3, are NP-hard. In this paper we study parameterized algorithms, i.e., algorithms that attempt to solve NP-hard problems by confining the combinatorial explosion to a parameter $k$. More precisely, a problem is {\em fixed-parameter tractable (FPT)} with respect to a parameter $k$ if it can be solved in time $O^*(f(k))$ for some function $f$, where $O^*$ hides factors polynomial in the input size.
Our results have also consequences in the field of moderately exponential-time algorithms, where one aims at providing an $O(c^n)$-time algorithm, with the constant $c>1$ being as small as possible.

We develop an FPT algorithm for {\sc $(k,l)$-Tree} in general graphs that relies upon a non-trivial generalization of the technique underlying the \probHamilton and \probkPath algorithms in \cite{bjorklund-hamilton,BHKK-narrow-sieves}.
As a result, we break the natural barrier of $O^*(2^n)$ in the running time bound for the {\sc $k$-IST} problem.
Then, we conduct a thorough examination of the {\sc $(k,l)$-Tree} problem in special classes of graphs that admit a vertex coloring with few colors.
This in turn implies faster algoritms for \probHamilton, \probkPath and \probkIST in bounded degree graphs.
Apart from the classic vertex coloring, we consider fractional coloring and vector coloring, thus showing that the latter tool, famous in approximation algorithms, is also helpful in parameterized complexity.

\subsection{Related Work}

\begin{table}[center]
\centering
\begin{tabular}{|l|c|c|l|c|}
	\hline
	Reference 		                               & Det./Rand. & Space & Graph $G$     		 & Running Time          \\\hline \hline		
	Prieto et al.~\cite{kISP24klogk}						 & {\em det} 	& {\em poly} & Undirected 						 & $O^*(2^{O(k\log k)})$ \\\hline	
	Gutin al.~\cite{kIOB2klogk}	      					 & {\em det} 	& {\em exp} & Directed 							 & $O^*(2^{O(k\log k)})$ \\\hline
	Cohen et al.~\cite{kIOB49k}		      			   & {\em det} 	& {\em exp} & Directed 							 & $O^*(55.8^k)$         \\		
	 														                 & {\em rand} & {\em poly} & Directed 							 & $O^*(49.4^k)$         \\\hline	
	Fomin et al.~\cite{kIOB16k}		      			   & {\em det} 	& {\em exp} & Directed 							 & $O^*(16^{k+o(k)})$		 \\
																							 & {\em rand} & {\em poly} & Directed 							 & $O^*(16^{k+o(k)})$		 \\\hline	
	Fomin et al.~\cite{kISP8k}		      			   & {\em det} 	& {\em poly} & Undirected 						 & $O^*(8^k)$            \\\hline	
	Shachnai et al.~\cite{esarepresentative} 		 & {\em det}  & {\em exp} & Directed 					     & $O^*(6.855^k)$        \\\hline	
	Daligault~\cite{thesis11}, Zehavi~\cite{zehavi-ipec13}& {\em rand}&{\em poly}& Directed 		 & $O^*(4^k)$            \\\hline
	Li et al.~\cite{4kist}	    							   & {\em det} 	& {\em poly} & Undirected 						 & $O^*(4^k)$            \\\hline
	Zehavi~\cite{zehavi-arxiv}	      					 & {\em det} 	& {\em exp} & Directed 							 & $O^*(5.139^k)$        \\
		 														               & {\em rand} & {\em exp} & Directed 							 & $O^*(3.617^k)$        \\\hline	
	{\bf This paper}								             & {\bf rand} & {\bf poly} & {\bf Undirected} 			 & $\bf O^*(3.455^k)$    \\
\hline	
\end{tabular}\medskip
\caption{Known FPT algorithms for {\sc $k$-IST}.}
\label{tab:knownresults1}
\end{table}

The {\sc $k$-IST} problem, along with its directed version known as {\sc $k$-Internal Out-Branching ($k$-IOB}\footnote{In $k$-IOB problem, we need to decide if a directed graph $G$ contains a spanning tree $T$ with exactly one vertex of in-degree 0 and at least $k$ internal vertices.}), are of interest in database systems \cite{outbranchpatent} and water supply network design~\cite{kISPbounddeg}. Note that any algorithm for $k$-IOB works also for $k$-IST (after replacing every edge of the input graph by two oppositely oriented arcs).
Over the last 10 years $k$-IST and $k$-IOB were heavily researched, resulting in a number of algorithms using a variety of approaches.
We present the progress for FPT algorithms in Table~\ref{tab:knownresults1}.
For the special case of graphs of bounded degree $\Delta$, Zehavi~\cite{zehavi-ipec13} shows an algorithm for $k$-IOB running in time $O^*(4^{(1-\frac{\Delta+1}{2\Delta(\Delta-1)})k})$.
Another specialized algorithm, due to Raible et al.~\cite{kISPbounddeg}, solves $k$-IST in graphs of maximum degree 3 in time $O^*(2.137^k)$.
There also has been quite some interest in moderately exponential time algorithms for $k$-IST.
Unlike in many other graph problems, even an $O^*(2^n)$ algorithm is not completely trivial.
Algorithms achieving this bound were shown by Raible et al.~\cite{kISPbounddeg} (with exponential space) and Nederlof~\cite{NedkIST} (polynomial space).
Note that $O^*(2^n)$ is a natural barrier, as it corresponds to the number of all subsets of vertices. During recent years researchers managed to surpass this barrier for a number of problems, like {\sc Dominating Set} by Grandoni~\cite{domset-breaking}, {\sc Feedback Vertex Set} by Razgon~\cite{fvs-breaking} and {\sc Hamiltonicity} by Bj\"{o}rklund~\cite{bjorklund-hamilton}.
Raible et al.~\cite{kISPbounddeg} explicitly posed an open question asking whether the approach of Bj\"{o}rklund~\cite{bjorklund-hamilton} can be extended to $k$-IST either for general graphs or for graphs of large vertex cover.
Raible et al.~\cite{kISPbounddeg} were able to cross the $O^*(2^n)$ barrier for graphs of bounded degree --- they get an algorithm running in time $O^*((2^{\Delta+1}\!-\!1)^{\frac{n}{\Delta-1}})$ and present a further improvement to $O^*(1.862^n)$ when $\Delta=3$.


The {\sc $(k,l)$-Tree} problem was implicitly introduced by Cohen et al.~\cite{kIOB49k} as a tool for solving $k$-IOB.
They obtained a $O^*(6.14^k)$-time algorithm. Currently, the fastest randomized algorithms, due to Daligault~\cite{thesis11} and Zehavi~\cite{zehavi-ipec13}, run in time $O^*(2^k)$ and have a polynomial space complexity.
Note that we meet the $O^*(2^k)$ barrier again. 
For the deterministic case, Zehavi~\cite{zehavi-arxiv} (building upon~\cite{esarepresentative}) shows a $O^*(2.597^k)$-time algorithm for {\sc $(k,l)$-Tree} with an exponential space complexity.

The fundamental \probkPath problem is well-studied in the field of Parameterized Complexity.
In the past three decades, it enjoyed a race towards obtaining the fastest FPT algorithm (see \cite{colorcoding,BHKK-narrow-sieves,Bodlaender93,divandcol,representative,productFam,koutis-icalp08,Monien85,williams-ipl,zehavi-arxiv}).
Currently, the best randomized algorithm, due to Bj\"{o}rklund et al.~\cite{BHKK-narrow-sieves}, runs in time $O^*(1.657^k)$ and has a polynomial space complexity, and the best deterministic algorithm, due to Zehavi \cite{zehavi-arxiv} runs in time $O^*(2.597^k)$ and has an exponential space complexity.

The result in \cite{BHKK-narrow-sieves} extends Bj\"{o}rklund's $O^*(1.657^n)$ time algorithm for \probHamilton in \cite{bjorklund-hamilton}.
The same paper contains also a better bound of $O^*(2^{n/2})$ for bipartite graphs.
\probHamilton in graphs of bounded maximum degree $\Delta$ received a considerable attention beginning with the paper of Eppstein~\cite{Eppstein07},  followed by works of Iwama and Nakashima~\cite{iwama}, Gebauer~\cite{Gebauer08} and Bj\"{o}rklund et al.~\cite{BHKK-TSP}.
Currently the best algorithm for $\Delta=3$, due to Cygan et al.~\cite{cut-count}, runs in time $O^*(1.201^n)$ \cite{cut-count}, while prior to this work for $\Delta\ge 4$ the best bound was that of the general Bj\"orklund's $O^*(1.657^n)$-time algorithm.


\subsection{Our Contribution}
We obtain improved algorithms for {\sc $(k,l)$-Tree}, {\sc $k$-IST}, \probkPath and \probHamilton. All of our algorithms are randomized Monte-Carlo and have polynomial space complexities. Our contribution is twofold.

\heading{From Paths to Trees}
First, we develop an algorithm that solves {\sc $(k,l)$-Tree} in general graphs in time $O^*(1.657^k2^{\ell/2})$.
This can be seen as a generalization of the technique of Bj\"{o}rklund et al.~\cite{bjorklund-hamilton,BHKK-narrow-sieves} from detecting paths to detecting trees.
In the original technique one enumerates walks of length $k$ (rather than paths, which are walks without repeating vertices) and then uses an algebraic tool to sieve-out the walks which are not paths.
The algebraic tool is to design a polynomial which is a sum of monomials, each corresponding to a walk.
The trick is that thanks to the use of bijective labellings, the non-path walks can be paired-up so that both corresponding monomials are the same and thus cancel over a field of characteristic 2.
It is quite clear that in the tree case walks should be replaced by branching walks (see~\cite{nederlof-steiner}).
However, the main difficulty lies in the pairing argument, which requires a new labelling scheme and becomes much more delicate.
Indeed, this extension is non-trivial as it is exactly the topic of the open problem posed by Raible et al.~\cite{kISPbounddeg} mentioned in the previous section.
Our algorithm, similarly as in~\cite{BHKK-narrow-sieves} uses a random bipartition of the vertices.
The running time depends in a crucial way on a random variable (called the number of needed labels) in the resulting probability space.
The second difficulty was to determine the distribution of this variable (Lemma \ref{lem:prob}), and again this turned out to be much more tricky than in the path case.

The {\sc $(k,l)$-Tree} algorithm described above already implies an improved running time for {\sc $k$-IST} in general graphs.
However, it works much faster if the hidden spannning tree has few leaves.
We design a different strategy, based on finding a maximum matching, to accelerate the algorithm when it looks for solutions with many leaves.
By merging the two strategies, we obtain an $O^*(3.455^k)$ time algorithm for {\sc $k$-IST}.
This is the first $O^*((4-\Omega(1))^k)$-time algorithm that uses polynomial space.
It immediately implies a moderately exponential time algorithm running in time $O^*(1.946^n)$, which breaks the $O^*(2^n)$ barrier.

\begin{table}[center]
\centering
\begin{tabular}{||c||c||c||c||}
  \hline
  $\Delta$ & Running Time \\
  \hline
  $3$ &  $O^*(1.5705^{k})$ \\
  \hline
$4$ &  $O^*(2^{2k/3})=O^*(1.5874^{k})$ \\
  \hline
  $5, 6$ &  $O^*(2^{7k/10})=O^*(1.6245^{k})$ \\
  \hline
  $7, 8$ &  $O^*(2^{5k/7})=O^*(1.6406^{k})$ \\
  \hline
  $9, 10$ & $O^*(2^{13k/18})=O^*(1.6497^{k})$\\
  \hline
  $11, 12$ & $O^*(2^{8k/11})=O^*(1.6555^{k})$ \\
  \hline
\end{tabular} \quad \begin{tabular}{||c||c||c||c||}
  \hline
  $\chi_v(G)$ & Running Time \\
  \hline
  4 & $O^*(1.6199^{k})$\\
  \hline
  5 & $O^*(1.6356^{k})$\\
  \hline
  6 & $O^*(1.6448^{k})$\\
  \hline
  7 & $O^*(1.6510^{k})$\\
  \hline
  8 & $O^*(1.6554^{k})$\\
  \hline
\end{tabular}\smallskip

\smallskip
\caption{Running times of our \probkPath algorithm in bounded degree graphs (left), and bounded vector chromatic number (right).
Replacing $k$ by $n$ gives bounds for \probHamilton.}
\label{tab:boundeddegree}
\end{table}

\heading{Paths and Trees in Colored Graphs}
Next, we study the {\sc $(k,l)$-Tree} problem in graphs that admit a vertex coloring with $d$ colors.
This can be seen as an extension of Bj\"{o}rklund's $O^*(2^{n/2})$-time algorithm for \probHamilton in bipartite graphs~\cite{bjorklund-hamilton}.
However, the insight in~\cite{bjorklund-hamilton} was to find a small vertex cover of the hidden solution.
Here, we use a different insight: by choosing roughly half of the color classes we get a small subset of the hidden solution vertices which covers many (but not all) of its edges.
The resulting algorithm runs in $\displaystyle{O^*(2^{(1 - \frac{\floor{{d}/2} \ceil{{d}/2}}{d(d-1)})k})}$ time when $l=O(1)$ (see Section~\ref{sec:proper} for a more complicated bound in the general case).
For a graph of bounded degree $\Delta$ that is neither complete nor an odd cycle, it is possible to construct a proper $\Delta$-coloring in linear time (e.g., by Lov\'asz's proof of Brooks' theorem \cite{deltacol}).
This immediately results in a fast algorithm for {\sc $(k,l)$-Tree} in bounded degree graphs (see Table \ref{tab:boundeddegree} for the special cases of \probkPath and \probHamilton; the case $\Delta=3$ is solved by a special algorithm, see below), along with an improved algorithm for {\sc $k$-IST} in such graphs (see Table \ref{tab:knownresults3} in Section~\ref{sec:bounddeg}).

Fractional coloring is a well studied generalization of the classical vertex coloring.
Our algorithm for graphs of low chromatic number generalizes quite easily to the case of low fractional chromatic number.
This has consequences in improved algorithms for some special graph classes, e.g. an $O^*(1.571^k 1.274^{l})$-time algorithm for {\sc $(k,l)$-Tree} in  subcubic triangle-free graphs, or $O^*(1.571^k)$-time algorithm for \probkPath in general subcubic graphs. For subcubic graphs of even larger girth, we get further improved bounds.

Another relaxation of the classical coloring is vector coloring, known for its importance in approximation algorithms (e.g.~\cite{GW}).
Its important advantage is that, contrary to the classical or fractional coloring, a $(1+\epsilon)$-approximation can be found in polynomial time~\cite{Gartner}.
We provide an algorithm for {\sc $(k,l)$-Tree} that applies vector coloring.
It results in improved running time when the vector chromatic number is at most $8$.

\subsection{Notation}

Throughout the paper we consider undirected graphs.
For an integer $k$, by $[k]$ we denote the set $\{1, 2, \ldots, k\}$.
For a set $S$ and an integer $k$ by ${S \choose k}$ we denote the family of all subsets of $S$ of size $k$.
Let us write $V=V(G)=[n]$ for the vertex set and $E=E(G)$ for the edge set of the host graph $G$.

\section{Finding Trees on $k$ Vertices with $l$ Leaves}
\label{sec:general}

In this section we generalize the \probkPath algorithm by Bj\"orklund et al.~\cite{bjorklund-hamilton,BHKK-narrow-sieves} to finding subtrees with $k$ nodes including $l$ leaves. We focus on undirected graphs.
Throughout this section we assume we are given a fixed partition $V=V_1 \cup V_2$ of the vertices of the input graph.
To get the promised generalization, in Section~\ref{sec:prob} we will use a {\em random} bipartition, similarly as in~\cite{bjorklund-hamilton,BHKK-narrow-sieves}.
In the next sections we investigate other partitions.

\subsection{Branching Walks}

The notion of branching walk was introduced by Nederlof~\cite{nederlof-steiner}.
A mapping $h:V(T)\rightarrow V(G)$ is a {\em homomorphism} from a graph $T$ to the host $G$ if for all $\{a,b\}\in E(T)$ it holds that $\{h(a),h(b)\}\in E(G)$.
We adopt the convention of
calling the elements of $V(T)$ {\em nodes} and the elements
of $V(G)$ {\em vertices}.

A {\em branching walk} in $G$ is a pair $B=(T,h)$ where $T$ is
an unordered rooted tree and $h:V(T)\rightarrow V(G)$ is a homomorphism
from $T$ to $G$.

Let $B=(T,h)$ be a branching walk in $G$.
The walk {\em starts} from the vertex $h(1)$ in $G$.
The {\em size} of the walk is $|V(T)|$.
The walk is {\em simple} if $h$ is injective.
The walk is {\em weakly simple} if for any node $x\in V(T)$, the homomorphism $h$ is injective on children of $x$.
The walk is {\em U-turn-free} if for any node $a$ of $T$, every child $b$ of $a$ maps to a different vertex than the parent $c$ of $a$, i.e., $h(b)\ne h(c)$.

A {\em proper order} of $B$ is any permutation $\pi:V(T)\rightarrow\{1,\ldots,|V(T)|\}$ such that for every two nodes $a,b\in V(T)$,
\begin{enumerate}[$(i)$]
\item if $\depth(a) < \depth(b)$ then $\pi(a) < \pi(b)$,
\item if $a,b\ne\treeroot(T)$ and $\pi(\parent(a)) < \pi(\parent(b))$ then $\pi(a) < \pi(b)$,
\item if $a$ and $b$ are siblings and $h(a) < h(b)$ then $\pi(a) < \pi(b)$.
\end{enumerate}
The following proposition is immediate.
\begin{proposition}
\label{prop:properly}
Every weakly simple branching walk has exactly one proper order.
\end{proposition}
We say that a weakly simple branching walk $B$ is properly ordered if $V(T)=\{1,\ldots,|V(T)|\}$ and the proper order from Proposition~\ref{prop:properly} is the identity function.

\subsection{Labeling}\label{sec:labeling}

For a tree $T$, by $L(T)$ we note the set of leaves of $T$ and by $I(T)$ we denote the set of internal vertices of $T$, i.e., $I(T)=V(T)\setminus L(T)$.

Like in~\cite{BHKK-narrow-sieves} or~\cite{BKK-motif-stacs} our crucial tool are {\em labelled} branching walks.
In~\cite{BKK-motif-stacs}, every node in the tree $T$ of a branching walk $(T,h)$ was assigned a label.
Here, similarly as in~\cite{BHKK-narrow-sieves}, we do not assign labels to some nodes, but we assign labels to some edges of $T$.
We define the set of {\em labellable elements} of a branching walk $B=(T,h)$ as

\[\la(B) = L(T) \cup (h(I(T)) \cap V_1) \cup \{uv \in E(T)\ :\ h(u),h(v) \in V_2\}. \]

Similarly, for a subtree $T$ of graph $G$, define

\[\la(T) = L(T) \cup (I(T) \cap V_1) \cup \{uv \in E(T)\ :\ u,v \in V_2\}. \]

We say that a branching walk $B=(T,h)$ is {\em admissible} when $B$ is weakly simple, U-turn-free, and properly ordered.
For nonnegative integers $k,l,r$, we also say that $B$ is {\em $(k,l,r)$-fixed} if $T$ has $k$ nodes and $l$ leaves, and $|\la(B)| = r$.

\subsection{The Polynomial}\label{sec:polynomial}

Let $r$ be an integer.
We use three kinds of variables in our polynomial.
First, for any edge $uv\in E(G)$, where $u<v$, we have a variable $x_{uv}$. For simplicity we will denote $x_{vu} = x_{uv}$.
Second, for each $q\in V\cup E$ and for each $l \in [r]$ we have a variable $y_{q, l}$.
Third, for each $v \in V(G)$, we have a variable $z_v$.
By $\x$ we denote the sequence of all $x_{uv}$-type variables, while by $\y$ we denote the sequence of all $y_{q,l}$-type variables.

For a branching walk $B=(T,h)$ and a labeling $\ell:\la(B)\rightarrow [|\la(B)|]$ we define the monomial

\[\mon(B,\ell) =z_{h(1)}\prod_{\substack{\{u,v\}\in E(T)\\ u<v}} x_{h(u),h(v)} \prod_{q\in \la(B)}y_{h(q),\ell(q)}\,,\]

where for $uv\in E(T)$, $h(uv)$ refers to edge $h(u)h(v)\in E(G)$. Next, we define a multivariate polynomial
$P_i$ with coefficients in a field of characteristic two by setting

\[
P_i = \sum_{\substack{B=(T,h)\\ \text{$B$ is admissible}\\ \text{$B$ is $(k,l,i)$-fixed}}}\ \sum_{\substack{\ell:\la(B)\rightarrow[|\la(B)|]\\\text{$\ell$ bijective}}} \mon(B,\ell)\,.
\]

Finally, let $P_{r\downarrow} = \sum_{i=2}^{r} P_i$.

\begin{lemma}
\label{lem:canceling}
The set of pairs of the form $(B,\ell)$, where $B$ is a non-simple, admissible and $(k,l,i)$-fixed branching walk and $\ell:\la(B)\rightarrow[|\la(B)|]$ is a bijection, can be partitioned into pairs, and the two monomials corresponding to each pair are identical.
\end{lemma}

\begin{proof}
Fix a pair $(B, \ell)$ from the set in the claim.

\mycase{1.} There exist two elements $e_1,e_2\in \la(B)$ such that $h(e_1)=h(e_2)$.
We take the pair with lexicographically minimal value of $(\min\{\ell(e_1),\ell(e_2)\},\max\{\ell(e_1),\ell(e_2)\})$; note that there is a unique minimum since $\ell$ is injective.

We define $\ell':\la(B)\rightarrow[|\la(B)|]$ as follows:

\[
\ell'(x) = \begin{cases}
\ell(u) & \mbox{if } x=v, \\
\ell(v) & \mbox{if } x=u, \\
\ell(x) &\mbox{otherwise.}
\end{cases}
\]

Note that $\ell' \ne \ell$ because $\ell$ is injective.
In this way, to the labelled branching walk $(B,\ell)$ we assigned another labelled branching walk $(B,\ell')$.
Note that $\mon(B,\ell) = \mon (B,\ell')$.
Observe that $(B,\ell')$ is also admissible and $(k,l,i)$-fixed, and $B$ is non-simple and falls into Case 1.
Moreover, if we begin from $(B,\ell')$ and follow the same way of assigning labels, we get $(B,\ell)$.
Thus we have paired up all the non-simple labelled branching walks that fall into Case 1.
In what follows we assume that Case 1 does not apply.

Observe that once Case 1 is excluded, no two edges of $T$ map to the same edge of $G$.
Indeed, assume that there is a pair of edges $u_1v_1$ and $u_2v_2$ of $T$ that map to the same edge $uv\in E(G)$, with $h(u_1)=h(u_2)=u$ and $h(v_1)=h(v_2)=v$.
Then, either one of $u$ or $v$ is in $V_1$ or both are in $V_2$.
By the definition of $\la(B)$, in the prior case $u_1,u_2\in\la(B)$, or $v_1,v_2\in\la(B)$ and in the latter case $u_1v_1,u_2v_2\in\la(B)$.
In both situations Case 1 applies.

Define \[A = \{(u,v)\ :\ \text{$u,v \in V(T)$ and $h(u)=h(v)$}\}.\]

Since Case 1 does not apply, and $B$ is not simple, we know that for every pair $(u,v)\in A$ at most one of $u$, $v$ is in $\la(B)$.
Let \[A^* = \{(u,v) \in A\ :\ h(u) = \min_{(x,y)\in A} h(x)\}.\]

Consider any node $x \in V(T)$. Then, we define

\[\ell^\downarrow(x) = \begin{cases}
                       \ell(x) & \text{if $x\in\la(B)$,} \\
                       \min\{\ell(z)\ :\ z \in \children(x) \cap \la(B)\} & \text{if $\children(x) \cap \la(B) \ne \emptyset$,} \\
                       \min\{\ell(xz)\ :\ z \in \children(x)\}  & \text{otherwise.}
                      \end{cases} \]

%
%
%

\mycase{2.} There is a pair $(u, v) \in A^*$ such that $u$ is an ancestor of $v$.
Let us define a subset $A^{**}\subseteq A^*$ consisting of pairs $(u,v)\in A^*$ such that node $u$ is an ancestor of $v$, and for every node $x$ on the path from $v$ to the root, if $h(x)=h(u)$ then $x\in\{u,v\}$.
We choose the pair $(u,v) \in A^{**}$ such that the pair $(u,\ell^\downarrow(v))$ is lexicographically minimal.
Note that there is exactly one such pair. Indeed, consider another pair $(u',v')\in A^{**}$ with $(u',\ell^\downarrow(v'))=(u,\ell^\downarrow(v))$.
Clearly, $u=u'$.
If $v\ne v'$ then, since $\ell^\downarrow(v)=\ell^\downarrow(v')$, it means that $v$ is a child of $v'$ or vice versa.
However, $h(v)=h(v')$, so $h$ is not a homomorphism, a contradiction.

Let $T'$ be the tree obtained from $T$ by reversing the path from $u$ to $v$.
More precisely, consider the path $p_{uv}=uu_1\cdots u_tv$.
Remove the edges $uu_1$ and $u_tv$ from $T$, obtaining three trees:
$T_u$ (rooted at $\treeroot(T)$),
$T_p$ (rooted at $u_1$), and
$T_v$ (rooted at $v$).
Next, root $T_p$ at $v_t$, obtaining $T_p'$.
Finally, make the root of $T_p'$ the child of $u$ in $T_u$, and make $v$ the child of $u_1$ in $T_p'$.
The resulting tree is denoted by $T'$.

Note that the multigraphs $h(T')$ and $h(T)$ are equal.
Now we essentially want to leave $h$ unchanged; however the branching walk $B'=(T',h)$ is not properly ordered and we need to renumber the nodes.

Note that $B'$ is weakly simple since no two edges of $T$ (and hence also of $T'$) map to the same edge of $G$.
Hence, by Proposition~\ref{prop:properly}, there is exactly one proper order $\pi$ of $B'$.
Let $T''$ be obtained from $T'$ by replacing every node $x$ by $\pi(x)$.
Let us define homomorphism $h''$ by $h''(x)=h(\pi^{-1}(x))$ for every $x\in V(T'')$, and let $B''=(T'',h'')$.
Note that $T''$ has the same number of leaves and internal vertices as $T$ and $B''$ is properly ordered.
Also, $T''$ is weakly simple and U-turn-free, for otherwise $T''$ has two edges that map by $h''$ to the same edge $e$ of $G$, which implies that $T$ has two edges that map by $h$ to $e$, a contradiction.
Hence $B$ is admissible.
Note that there is a one-to-one correspondence between elements of $\la(B)$ and $\la(B'')$.
It follows that $B''$ is $(k,l,i)$-fixed.

Now let us define a labeling $\ell''$ of $B''$.
For every node $x\in \la(B'')$, we put $\ell''(x) = \ell(\pi^{-1}(x))$.
Note that $h(T)=h(T'')$ as multigraphs, and hence since in $T$ no two edges are mapped by $h$ to the same edge of $G$, the same holds for $T''$.
It follows that for every edge $xy$ of $T''$ there is a unique edge $x_0y_0$ of $T$ such that $h(xy)=h(x_0y_0)$.
For every edge $xy\in E(T'')\cap\la(B'')$ we pick the corresponding edge $x_0y_0$ of $T$, and we put $\ell''(xy)=\ell(x_0y_0)$.
The definitions of $\ell''$ and $h''$ imply that $\mon(B,\ell)=\mon(B'',\ell'')$.

We claim that $B''\ne B$.
Note that $v \ne u_1$, for otherwise $h$ is not a homomorphism.
Hence, the path $p_{uv}=uu_1\cdots u_tv$ has at least one internal vertex, i.e., $t\ge 1$.
Consider the walks $W=h(u),h(u_1),\ldots,h(u_t),h(v)$ and $W''=h(v),h(u_t),\ldots,h(u_1),h(u)$ in $G$.
If $B''=B$ then, in particular, $W=W''$. But this implies that $T$ is not U-turn-free, a contradiction.

Hence to $(B,\ell)$ we have assigned a different pair $(B'',\ell'')$, which also satisfies Case 2 but does not satisfy Case 1.
Let $(u',v')$ be the pair of nodes chosen in Case 2 for $(B'',\ell'')$.
Note that in $(B'',\ell'')$ we have $h''(u)=h''(v)=h''(u')=h''(v')$.
Since only the subtree of $u$ is changed, for every node $x\le u$ we have $\pi(x)=x$.
It follows that $u'=u$.
Note that $\ell''^\downarrow$ can differ from $\ell^\downarrow$ only for the vertices $u_1,\ldots,u_t$ on the path $p_{uv}$.
However, none of them maps to $h(u)$, because otherwise $(u,v)\not\in A^{**}$.
Hence, $v'=v$. Thus we have paired up all the non-simple, admissible and $(k,l,i)$-fixed labelled branching walks that fall into Case 2, but do not fall into Case 1.

\mycase{3.} There is a pair $(u,v) \in A^*$ such that neither $u$ is an ancestor of $v$ nor $v$ is an ancestor of $u$.
Assume w.l.o.g.\ that $h(\parent(u)) \le h(\parent(v))$.
If there are many such pairs $(u,v)$, we take the pair with lexicographically minimal value of $(h(\parent(u)),h(\parent(v)))$.
We claim there is exactly one such pair.
Indeed, consider another pair $(u',v')$. Then $h(u)=h(u')=h(v)=h(v')$ since both pairs are in $A^*$.
Also $h(\parent(u))=h(\parent(u'))$ and $h(\parent(v))=h(\parent(v'))$ by the minimality.
Then $u=u'$ for otherwise two edges of $T$ map to the same edge $h(u)h(\parent(u))$ of $G$.
Similarly, $v=v'$ for otherwise two edges of $T$ map to the same edge $h(v)h(\parent(v))$ of $G$.

Let $T_u$ and $T_v$ be the subtrees of $T$ rooted at $u$ and $v$, respectively.
Let $T'$ be the tree obtained from $T$ by swapping $T_u$ and $T_v$.
Note that $h(T')=h(T)$ as multigraphs.
Now we define a branching walk $B''=(T'',h'')$ by renumbering vertices of $T'$ using a proper order $\pi$, exactly as in Case 2.
By the same argument as in Case 2, $B''$ is admissible and $(k,l,i)$-fixed.
We also define a labeling $\ell''$ of $B''$ exactly in the same way as in Case 2.
It follows that $\mon(B,\ell) = \mon(B'',\ell'')$.

We claim that $B''\ne B$. Otherwise, there is an isomorphism $i:V(T_u)\rightarrow V(T_v)$ such that for every node $x\in V(T_x)$ we have $h(x)=h(i(x))$.
Pick any leaf $z\in V(T_u)$. Then $i(z)$ is also a leaf. Hence both $z$ and $i(z)$ are in $\la(B)$.
It follows that Case 1 applies to $z$ and $i(z)$, a contradiction.

It is clear that $(B'',\ell'')$ also satisfies Case 3 and it does not satisfy Case 1.
Assume it satisfies Case 2, i.e., there are nodes $(u^*, v^*) \in A^*$ such that $u^*$ is an ancestor of $v^*$.
Note that $h(u^*)=h(v^*)=h(u)=h(v)$.
We can assume that $u^*=u$ or $u^*=v$, for otherwise the pair $(u, v)$ satisfies Case 2 for the labelled branching walk $(B,\ell)$, a contradiction.
By symmetry assume $u^*=u$.
Then $\{v,v^*\}$ satisfies Case 2 for the labelled branching walk $(B,\ell)$, a contradiction again.
Hence to $(B,\ell)$ we have assigned a different pair $(B'',\ell'')$, which also satisfies Case 3 but does not satisfy the earlier cases.

We claim that if we begin from $(B'',\ell'')$ and follow the same way of assignment we get $(B,\ell)$ back.
Indeed, note that only the subtrees of $u$ and $v$ are changed, but since Case 2 is excluded, these subtrees contain no vertices that map to $h(u)$, except for $u$ and $v$. It follows that in $(B'',\ell')$ the pair $(\pi(u),\pi(v))$ is chosen.
After swapping the subtrees (and renumbering the vertices, redefining the labelling, etc.) we get $(B,\ell)$ back.

Thus we have paired up all the non-simple, admissible and $(k,l,i)$-fixed labelled branching walks that fall into Case 3, but do not fall into the earlier cases.
This finishes the proof.
\end{proof}

\begin{lemma}
\label{lem:polynomial}
The polynomial $P_{\downarrow r}$ is non-zero iff the input graph contains a subtree $T_G$ with $k$ nodes and $l$ leaves, such that $\la(T_G)\le r$.
\end{lemma}

\begin{proof}
First assume $P_{\downarrow r}$ is non-zero.
It follows that $P_{\downarrow r}$ contains at least one monomial with non-zero coefficient.
By Lemma~\ref{lem:canceling}, every such monomial corresponds to a simple branching walk $B=(T,h)$, because all the monomials corresponding to non-simple branching walks cancel-out over a field of characteristic two. Hence $h(T)$ is the desired tree.

Let $T_G$ be a subtree of $G$ with $k$ nodes and $l$ leaves such that $\la(T)\le r$.
Pick an arbitrary vertex $v$ of $T_G$ and let $T_v$ be tree $T_G$ rooted at $v$.
By Proposition~\ref{prop:properly}, there is a proper order $\pi$ of the simple branching walk $(T_v,\id_{V(T_v})$.
Let $T$ denote the tree $T_v$ after replacing every vertex $x$ by $\pi(x)$.
Define a homomorphism $h:V(T)\rightarrow V(G)$ by putting $h(x)=\pi^{-1}(x)$.
Then $B=(T,h)$ is admissible.
Note that there is a one-to-one correspondence between $\la(T_G)$ and $\la(B)$.
It follows that $B$ is $(k,l,|\la(T_G)|)$-fixed.

Finally, choose an arbitrary bijection $\ell:\la(B)\rightarrow[\la(B)]$.
We must now have $P\not\equiv 0$ because we can uniquely reconstruct the pair $(B,\ell)$ from its monomial representation
$\mon(B,\ell)$: indeed, first recover the root vertex $v$ from the unique $z_v$-type indeterminate; second, recover $B=(T,h)$ from the $x_{a,b}$-type indeterminates using the fact that $B$ is simple, properly ordered, and starts from $v$; then recover $\ell$ from the $y_{x,l}$-type indeterminates using the fact that $\ell$ is injective.
\end{proof}

\subsection{Evaluating the Polynomial $P_{\downarrow r}$ in Time $O^*(2^r)$}
\label{sec:eval}

In this section we show that the polynomial $P_{\downarrow r}$ can be evaluated in a given point $(\x,\y,\z)$ in time $O^*(2^r)$.
Clearly, it suffices to show this bound for every polynomial $P_i$.
To this end, let us rewrite $P_i$ as a sum of $2^r$  polynomials such
that each of them can be evaluated in time polynomial in the input size.
For each $X\subseteq [r]$, let

\[
P_i^X= \sum_{\substack{B=(T,h)\\ \text{$B$ is admissible}\\ \text{$B$ is $(k,l,i)$-fixed}}}\ \sum_{\ell:\la(B)\rightarrow X} \mon(B,\ell)\,.
\]

Note that we do not assume that the labelings in the second summation are bijective.

\begin{lemma}
\label{lem:ie}
 $\displaystyle P_i = \sum_{X \subseteq{\{1,2,\ldots,k\}}} P_i^X.$
\end{lemma}

\begin{proof}
Let us fix an admissible  $(k,l,i)$-fixed branching walk $B=(T,h)$.
Note that a function $\ell:\la(B)\rightarrow [|\la(B)|]$ is bijective if and only if it is surjective, so
 \begin{equation}
 \label{eq:aa}
 \sum_{\substack{\ell:\la(B)\rightarrow [|\la(B)|]\\ \text{$\ell$ bijective}}} \mon(B,\ell) = \sum_{\substack{\ell:\la(B)\rightarrow [|\la(B)|]\\ \text{$\ell$ surjective}}}\mon(B,\ell).
 \end{equation}
Recalling the coefficients of $P_i$ are from a field of characteristic 2, and hence $-1=1$,
we have, by the Principle of Inclusion and Exclusion,
%
%
%
\begin{equation}
\label{eq:bb}
\sum_{\substack{\ell:\la(B)\rightarrow [|\la(B)|]\\ \text{$\ell$ surjective}}}\mon(B,\ell) = \sum_{X\subseteq [r]}\ \sum_{\ell:\la(B)\rightarrow  X}\ \mon(B,\ell)\,.
\end{equation}
From~\eqref{eq:aa} and~\eqref{eq:bb} we immediately obtain
\begin{equation}
\label{eq:cc}
P_i =  \sum_{\substack{B=(T,h)\\ \text{$B$ is admissible}\\ \text{$B$ is $(k,l,i)$-fixed}}}\ \sum_{X\subseteq [r]}\ \sum_{\ell:\la(B)\rightarrow  X}\ \mon(B,\ell)\,.
\end{equation}
The claim follows by changing the order of summation.
\end{proof}

Now we are left with a tedious job of evaluating $P_i^X(\x,\y,\z)$ in polynomial time.
To simplify notation in the running time bounds, let us write $m$ for the number of edges in $G$, and $\mu=O((k+r)\log (k+r)\log\log (k+r))$ for the time needed to multiply or add two elements of $\mathbb{F}_{2^{\ceil{\log((k+r))}+O(1)}}$.

\begin{lemma}
\label{lem:eval:X}
Given a nonempty $X\subseteq [r]$ and
three vectors $\x,\y,\z$ of values in $\mathbb{F}_{2^{\ceil{\log(k)}+O(1)}}$
as input, all the values of $P_i^X(\x,\y,\z)$ for $i=0,\ldots,r$ can be computed by
dynamic programming in time $O(mk^4l^2\mu)$ and space $O(mk^2l)$.
\end{lemma}

\begin{proof}
W.l.o.g.\ we assume that $V(G)=\{1,2,\ldots,n\}$.
For a vertex $a\in V(G)$, denote the ordered sequence of neighbors of $a$ in $G$ by $a_1<a_2<\cdots<a_{\deg_G(a)}$.
We say that a U-turn-free branching walk $B=(T,h)$ starting in a vertex $a\in V(G)$ is {\em $j$-forced} when for any child $u$ of the root $1$ in $T$ it holds that $h(u)\ge a_j$.


Our objective is to compute a five-dimensional array $A^X$
whose entries are defined by
\[
A^X[a,j,k',l',i] = \sum_{B=(T,h)} \sum_{\ell:\la(B)\setminus\{1\}\rightarrow X} \prod_{\substack{\{u,v\}\in E(T)\\ u<v}} x_{h(u),h(v)} \prod_{q\in \la(B)\setminus\{1\}}y_{h(q),\ell(q)}\,,
\]
where the outer sum is over all branching walks $B$ which are admissible, $(k',l',i)$-fixed, $j$-forced and start from $a$.
The entries of $A^X$ admit the following recurrence.
For $j=\deg_G(a)+1$, or $k'=1$, or $l'=0$, or $i < 0$ we have
\[
A^X[a,j,k',l',i]=
\begin{cases}
1 & \text{if $k'=1$, $l'=0$ and $i=[a\in V_1]$,}\\
0 & \text{otherwise}.
\end{cases}
\]
Define
\[
f(u,v,b)=
\begin{cases}
x_{u,v} \sum_{t_1,t_2\in X} y_{v,t_1} y_{uv, t_2} & \text{if $u,v\in V_2$ and $b=1$,}\\
x_{u,v} \sum_{t\in X} y_{uv, t}                   & \text{else if $u,v\in V_2$,}\\
x_{u,v} \sum_{t\in X} y_{v, t}                    & \text{else if $v\in V_1$ or $b=1$,}\\
x_{u,v}                                           & \text{otherwise}.
\end{cases}
\]
For $1\leq j\leq \deg_G(a)$, $2\leq k' \leq k$, $1\leq l' \leq l$, $0 \le i \le r$ we have
\begin{equation}
\label{eq:dp2}
\begin{aligned}
 A^X[a,j,k',l',i] =\; & A^X[a,j+1,k',l',i] + &&\\
                    & f(a,a_j,1) \cdot A^X[a,j+1,k'-1,l'-1,i-1-[a,a_j\in V_2]] + &&\\
                    & f(a,a_j,0) \cdot \!\!\!\!\!\!\!\! \sum_{\substack{k_1+k_2=k'\\k_1,k_2\ge 1}} \sum_{\substack{l_1+l_2=l'\\l_1 \ge 0\\ l_2\ge 1}} \sum_{\substack{i_1+i_2=i-[a,a_j\in V_2]\\i_1\ge 0\\i_2\ge 1}} \!\!\!\!\!\!\!\!\!\!\!\!\!\! A^X[a,j+1,k_1,l_1,i_1]\cdot A^X[a_j,1,k_2,l_2,i_2].
\end{aligned}
\end{equation}
To see that the recurrence is correct, observe that
the three lines above correspond to branching walks which are admissible, $(k',l',i)$-fixed, $j$-forced and start from $a$, where
either
(a) the root has no child which maps to $a_j$,
or
(b) the root has exactly one child which maps to $a_j$ and the child is a leaf,
or
(c) the root has exactly one child which maps to $a_j$ and the child is an internal node.
(At most one such child may exist because the branching walk is weakly simple.)

To recover the value of the polynomial $P^X(\x,\y,\z)$, we observe that
\begin{equation}
\label{eq:dp1}
P_i^X(\x,\y,\z)=\sum_{v\in V_1} z_v \sum_{t\in X}\ y_{v,t}\cdot A^X[v,1,k,l,i] + \sum_{v\in V_2} z_v \cdot A^X[v,1,k,l,i]\,.
\end{equation}

The time bound follows by noting that the values of $f(u,v)$ can be precomputed and tabulated in $O(k(n+m)\mu)$ time to accelerate the computations for the individual entries of $A^X$.
\end{proof}

From Lemmas~\ref{lem:ie} and \ref{lem:eval:X}, we immediately obtain the following.

\begin{lemma}
\label{lem:eval}
The polynomial $P_{\downarrow r}$ can be evaluated in time $O^*(2^r)$ and polynomial space.
\end{lemma}

\subsection{A Single Evaluation Algorithm}

Assume we choose the value of parameter $r$ large enough so that if there is a $(k,l)$-tree $T_G$ in $G$, then $|\la(T_G)|\le r$.
Then, by Lemma~\ref{lem:polynomial}, we can test the existence of a $(k,l)$-tree by testing whether the polynomial $P_{\downarrow r}$ is non-zero.
The latter task can be performed efficiently using a single evaluation of the polynomial $P_{\downarrow r}$.
To show that this is the case we need the Schwartz-Zippel Lemma, shown independently by DeMillo and Lipton~\cite{DeMilloLipton1978}, Schwartz~\cite{schwartz} and Zippel~\cite{zippel}.

\begin{lemma}
\label{lem:schwarz-zippel}
Let $p(x_1, x_2, \ldots, x_n)\in F[x_1,\dots,x_n]$ be a polynomial of degree at most $d$ over a field $F$, and assume $p$ is not identically zero.
Let $S$ be a finite subset of $F$. Sample values $a_1,a_2,\ldots,a_n$ from $S$ uniformly at random.
Then, \[\Pr{p(a_1,a_2,\ldots,a_n)=0}\le d/|S|.\]
\end{lemma}

\begin{lemma}
\label{lem:single-eval-algorithm}
Let $V(G)=V_1 \cup V_2$ be a fixed bipartition of the vertex set of the host graph $G$.
There is an algorithm running in $O^*(2^r)$ time and polynomial space such that
\begin{itemize}
 \item If $G$ does not contain a $(k,l)$-tree then the algorithm always answers NO,
 \item If $G$ contains a $(k,l)$-tree $T_G$ such that $|\la(T_G)|\le r$, then the algorithm answers YES with probability at least $\frac{1}2$.
\end{itemize}
\end{lemma}

\begin{proof}
 The algorithm is as follows: using the algorithm from Section~\ref{sec:eval}, we evaluate the polynomial $P_{\downarrow r}$ over the field $\field{2^{\ceil{\log_2 (k+r)}+1}}$, substituting the indeterminates by independently chosen random field elements.
 The time bound follows from Lemma~\ref{lem:eval}.

 If there is no $(k,l)$-tree in the input graph, by Lemma~\ref{lem:polynomial} the evaluation returns 0, so we report the correct answer.

 Now assume there is a $(k,l)$-tree $T_G$ such that $\la(T_G)\le r$.
 Then, by Lemma~\ref{lem:polynomial}, $P$ is a non-zero polynomial.
 Note that $\deg(P_i) = k + i$, hence $\deg(P_{\downarrow r}) \le k + r$.
 Hence, by the Schwartz-Zippel Lemma $P$ evaluates to the zero field element with probability at most $\frac{1}2$.
 This finishes the proof.
\end{proof}

\subsection{The random bipartition algorithm}\label{sec:prob}

In this section we assume that $V=V_1\cup V_2$ is a random bipartition, i.e., every vertex goes to $V_1$ independently with probability $1/2$.
Our plan is to choose the value of parameter $r$ large enough so that if there is a $(k,l)$-tree $T_G$ in $G$, then with high probability $|\la(T_G)|\le r$.
Then, by Lemma~\ref{lem:single-eval-algorithm}, we are done.
Of course, putting $r=k$ would perfectly achieve the above goal, but then we only get the running time of $O^*(2^{k})$, matching that of Zehavi~\cite{zehavi-ipec13}.

A natural choice is to set the value of $r$ close to the expectation of $|\la(T_G)|$.
The following lemma follows from the definition of $\la(T_G)$, by the linearity of expectation.

\begin{lemma}
For every $(k,l)$-tree $T_G$ in $G$, we have $\Expect{|\la(T_G)|} = \frac{3}4k + \frac{1}2l - \frac{1}{4}$.
\end{lemma}

By the lemma above and Markov's inequality, if we put $r = \frac{3}4k + \frac{1}2l$, then the probability that $|\la(T_G)|\le r$ is $\Omega(\frac{1}{k+l})$. Hence it suffices to repeat the algorithm from Lemma~\ref{lem:single-eval-algorithm} (i.e., evaluate the polynomial $P_{\downarrow r}$) $O(k+l)$ times answering true iff at least one evaluation was non-zero, to get a Monte-Carlo algorithm for testing the existence of a $(k,l)$-tree. The complexity of this algorithm is $O^*(2^{(3k+2l)/4})$.
However, similarly as in~\cite{bjorklund-hamilton,BHKK-narrow-sieves}, we can do better.
The idea is to use a value of $r$ smaller than that appearing in the expectation by an $\Omega(k)$ term.
Then the probability that a $(k,l)$-subtree is admissible is inverse-exponential.
Hence, we need to repeat the algorithm from Lemma~\ref{lem:single-eval-algorithm} exponentially many times, every time for a different random bipartition.
However, it turns out that for carefully selected values of $r$, this pays off.
To find this value, the following lemma is crucial.

\begin{lemma}
\label{lem:prob}
Fix an arbitrary $(k,l)$-tree $T_G$ in $G$.
For any integer $t$ such that $0 \le t \le (k-1) / 2$, we have
\[\Pr{|\la(T_G)| \le k + \tfrac{l}2 - t} \ge \frac{1}{2^{k+1}}{k-1 \choose 2t}.\]
\end{lemma}

\begin{proof}
 Root $T_G$ at an arbitrary vertex $r$.
 Let the random variable $X_{22}$ denote the number of edges $uv\in E(T_G)$ such that $u,v\in V_2$.
 Also, let $X_{1,i}$ denote the number of internal vertices in $V_1$.
 Then, by the definition of $\la(T_G)$, we have $|\la(G_T)|  = l + X_{1,i} + X_{22}$.

 Fix a subset of edges $S \in {E(T_G) \choose 2t}$.
 For $a=0,1$, let $c_a : V(T_G) \rightarrow \{1,2\}$ be the assignment of vertices of $T_G$ to sets $V_1$, $V_2$ such that for every $v\in V(T_G)$, we have $c_a(v)=1$ if and only if on the path from $r$ to $v$ in $T_G$ the number of edges from $S$ is congruent to $a$ modulo 2.
 Since every vertex is colored 1 in exactly one of the colorings $c_0$, $c_1$, we infer that
 \[X_{1,i}(c_0) + X_{1,i}(c_1) = k - l.\]
 Similarly, every edge in $E(T_G) \setminus S$ is colored 22 in exactly one of the colorings $c_0$, $c_1$; hence
 \[X_{22}(c_0) + X_{22}(c_1) = k - 1 - 2t.\]
 It follows that
 \[
 \begin{split}
\min\{X_{1,i}(c_0) + X_{22}(c_0), X_{1,i}(c_1) + X_{22}(c_1)\} & \le (k-l + k-1-2t)/2 \\
& < k - \tfrac{l}2 - t.
 \end{split}
\]
 Hence, for at least one of the colorings $c_0$, $c_1$, we have $|\la(T_G)| < k + \tfrac{l}2 - t$.
 For all choices of $S$ there are at least $\tfrac{1}2{k-1 \choose 2t}$ such colorings, so the claim follows.
\end{proof}

The following lemma follows immediately from Stirling's approximation.

\begin{lemma}
\label{lem:binom}
 For any fixed $\alpha$, $0 < \alpha < 1$,
 \[{n \choose \alpha n} = O^*\left(\left(\frac{1}{\alpha^{\alpha}(1-\alpha)^{1-\alpha}}\right)^n\right)\]
\end{lemma}

\begin{theorem}
\label{th:algorithm-random-partition}
 There is an $O^*(1.66^k 2^{l/2})$-time Monte-Carlo polynomial space algorithm for finding a $(k,l)$-tree which never reports a false positive, and reports false negatives with constant probability.
\end{theorem}

\begin{proof}
 Fix an $\epsilon \ge 0$ and let $t= \floor{(\frac{1}4+\epsilon) k}$.
 Put $r=k - t + \ceil{\frac{l}2} = \ceil{(\frac{3}4 - \epsilon)k} + \ceil{\frac{l}2}$.
 We choose a random bipartition $V=V_1\cup V_2$ and apply the algorithm from Lemma~\ref{lem:single-eval-algorithm}.
 We repeat this $\ceil{2^{k+1}/{k-1 \choose 2t}}$ times, and we return YES iff at least one of the executions of the algorithm from Lemma~\ref{lem:single-eval-algorithm} returned YES.
 If there is no $(k,l)$-tree in the input graph, by Lemma~\ref{lem:single-eval-algorithm} we report the correct answer.
 Now assume there is a $(k,l)$-tree $T_G$.
 Call a bipartition {\em nice} if $|\la(T_G)|\le r$.
 By Lemma~\ref{lem:prob}, a random bipartition is nice with probability at least $p=\frac{1}{2^{k+1}}{k-1 \choose 2t}$.
 Hence, at least one of the tried bipartitions is nice with probability at least $1-(1-p)^{1/p} \ge 1-1/e$.
 For such a bipartition, the algorithm from Lemma~\ref{lem:single-eval-algorithm} answers YES with probability at least $\frac{1}2$.
 Hence our algorithm reports a false-negative with probability at most $1/e + \frac{1}2 < 1$.

 By Lemma~\ref{lem:single-eval-algorithm}, the running time is
 \[O^*\left(2^{r+k}/{k-1 \choose 2t}\right) = O^*\left(2^{(7/4-\epsilon)k+l/2} / {k \choose (1/2+2\epsilon)k}\right).\]
 By Lemma~\ref{lem:binom} we can express this by $O^*((f(\epsilon))^k2^{l/2})$, for $f(\epsilon)=2^{7/4-\epsilon}(\frac{1}2+2\epsilon)^{\frac{1}2+2\epsilon}(\frac{1}2-2\epsilon)^{\frac{1}2-2\epsilon}$.
 The function $f$ attains a minimum smaller than $1.65685$ for $\epsilon = 0.042894$. Hence the claim.
\end{proof}

\section{Spanning Trees with Many Internal Vertices}\label{sec:kIST}
In this section we give an improved algorithm for {\sc $k$-IST}. As previous algorithms (see, e.g., \cite{kIOB49k}) for this problem, we depend on the following insight.

\begin{lemma}
A graph $G$ contains a spanning tree with at least $k$ internal vertices iff $G$ contains a subtree with exactly $k$ internal vertices and at most $k$ leaves.
\end{lemma}

We call such a tree a witness for $k$-IST. A witness tree is \emph{minimal} if one cannot remove any leaf from it without reducing the number of internal nodes.
It is easy to see that such a tree can be extended to a $k$-IST simply by connecting each remaining vertex in the graph to the tree arbitrarily.

We will use two strategies to search for a minimal witness tree. In the first, which we refer to as strategy A,  we search for a $(k+l,l)$-tree for every $l=2,\ldots,\alpha k$, for some $\alpha$ to be specified later, with our algorithm from Theorem~\ref{th:algorithm-random-partition}. Its runtime gets worse the more nodes, and especially leaves, the tree has. Running only strategy A to solve for a $k$-IST (i.e. let $\alpha=1$) results in an $O^*(3.883^k)$ time algorithm using polynomial space, already improving over state of the art.

In our second strategy, strategy B, we search indirectly for a $(k+l,l')$-tree for every $l=\alpha k,\ldots, k$, where $l'\leq l$, by using an observation from Zehavi~\cite{zehavi-arxiv}, that a tree with many leaves can be trimmed to a small tree with few leaves and a disjoint matching. Furthermore, such a disconnected subgraph can easily be combined into a $(k+l,l')$-tree for some $l'\leq l$ by connecting the edges of the matching one-by-one to the tree in an arbitrary order.

For completeness, we state and prove the following lemma implicit in~\cite{zehavi-arxiv}:

\begin{lemma}
Every minimal $(k+l,l)$-tree with $l>(k/2)+1$ has a subgraph consisting of a minimal $(k+l-2,l-1)$-tree and a disjoint edge.
\end{lemma}

\begin{proof}
Let us root the tree at an arbitrary leaf.
Every leaf's parent has only one child, otherwise it will not be minimal since we can remove the leaf.
Associate with a leaf $\ell$  the path $p(\ell)$ of all internal vertices which have only the leaf $\ell$ as their descendant among all
leaves. Note that every internal node is associated with at most one leaf,
and that $p(\ell)$ contains at least the parent (since the tree is minimal). If there is no leaf $\ell$ with
exactly one internal node in $p(\ell)$, we would have at least twice as many
internal nodes as leaves violating the assumption $l>(k/2)+1$. Hence there is a
leaf $\ell$ with $|p(\ell)|=1$, and we can remove $\ell,p(\ell)$ as the disjoint edge from
the tree, leaving a minimal $(k+l-2,l-1)$-tree.
\end{proof}

By iteratively applying the lemma above, we have that
\begin{corollary}
If a $k$-IST instance graph $G$ has a $((1+\beta)k,\beta k)$-tree witness for some $\beta>0.5$, then it also has a minimal $(3(1-\beta)k,(1-\beta)k)$-tree and a disjoint $(2\beta-1)k$ edge matching.
\end{corollary}

In contrast to~\cite{zehavi-arxiv}, that uses representative sets to separate the matching from the tree for {\sc $k$-IOB} (the directed version of the {\sc $k$-IST} problem), we will simply use a random bipartition.
It turns out to be a faster algorithm because in our algorithm for $(k+l,l)$-tree from Theorem~\ref{th:algorithm-random-partition}, the leaves cost much more than internal nodes, unlike in the directed counterpart from~\cite{zehavi-ipec13} where all nodes cost the same.

For each $l=\alpha k,\ldots, k$, we let $\gamma=l/k$ and repeat the following $r(\gamma)$ times.
We partition the vertices in $V_M\cup V_T=V$. With probability $p(\gamma)$ we put a vertex in $V_M$, otherwise we put it in $V_T$.
First we compute in polynomial time a maximum matching in the induced graph $G[V_M]$. If we find one with at least $(2\gamma-1)k$ edges, then we also search for a $(3(1-\gamma)k,(1-\gamma)k)$-tree in the induced graph $G[V_T]$ with the algorithm from Theorem~\ref{th:algorithm-random-partition}.
If we also find a tree, we combine the found matching and the tree into a $(k+l,l')$-tree witness and output {\em yes}.
If all runs for all $l$ fail we output {\em no}.
We set the parameters $r(\gamma)$ and $p(\gamma)$ as
\[
r(\gamma)=\left(\frac{1+\gamma}{4\gamma-2}\right)^{(4\gamma-2)k}\left(\frac{1+\gamma}{3(1-\gamma)}\right)^{3(1-\gamma)k}, \,\,\, p(\gamma)=\frac{4\gamma-2}{1+\gamma}.
\]
The probability that none of the $r(\gamma)$ tries succeeds in separating the tree and the matching is at most
\[
\left(1 - \left(\frac{4\gamma-2}{1+\gamma}\right)^{(4\gamma-2)k}\left(\frac{3(1-\gamma)}{1+\gamma}\right)^{3(1-\gamma)k} \right)^{r(\gamma)} = (1-r(\gamma)^{-1})^{r(\gamma)}<e^{-1}.
\]
We get that with probability at least $1-e^{-1}$ we successfully separate the matching and the tree in the bipartition, and hence strategy B works.
The running time for a fixed ratio $\gamma$ is $O^*(r(\gamma)1.657^{3(1-\gamma)k}2^{(1-\gamma)k/2})$.
One can check that this is a decreasing function for $\gamma \in [0.68,1]$.
Hence, when $\alpha \in [0.68,1]$, the total time of strategy B is $O^*(r(\alpha)1.657^{3(1-\alpha)k}2^{(1-\alpha)k/2})$.
On the other hand, strategy A takes time $O^*(1.657^{(1+\alpha)k}2^{\alpha k/2})$, which increases with $\alpha$.
The best trade-off between strategy A and B is obtained for $\alpha=0.8627...$, which gives a total running time of $O^*(3.455^k)$.

For $\alpha=0.8627$, in both strategies, the worst running time bound is obtained when $l=\alpha k$ {\em and} $k+l=n$; indeed, if $k+l>n$, it is clear that there is no solution. Therefore, the worst running time bound is obtained when $k=n/(1+\alpha)$, which makes our algorithm run in time $O^*(3.455^{n/(1+\alpha)})=O^*(1.946^n)$. Thus, we have proved the following result.

\begin{theorem}
There is an $O^*(\min\{3.455^k,1.946^n\})$-time Monte-Carlo polynomial space algorithm for {\sc $k$-IST} in general graphs.
\end{theorem}

\section{Colored Graphs}

In this section we improve the running times of algorithms from Section~\ref{sec:general} in restricted settings.
This is done by adjusting the partition $V=V_1 \cup V_2$ to particular graph classes where vertex colorings guide us in making the partition. We will consider three ways of coloring the vertices. The first is ordinary proper vertex coloring of the graph, i.e. color the vertices so that no edge has both its endpoints colored by the same color. The least number of colors needed is denoted by $\chi_G$.
The second way is fractional vertex coloring, that assigns a subset of $b$ colors to each vertex from a palette of $a$ colors so that the endpoints of each edge receive disjoint subsets of colors.
The smallest possible ratio $a/b$ is denoted by $\chi_f(G)$. The third way is vector coloring that assigns unit length vectors to the vertices so the minimum angle $\alpha$ between every edge's endpoints' vectors is as large as possible. The smallest possible value of $1+cos^{-1}(\alpha)$ is denoted by $\chi_v(G)$.

The following chain of inequalities holds (see e.g.~\cite{GvozdenovicL08}), where $\omega(G)$ is the clique number, i.e., the size of the largest clique in $G$.

\begin{theorem}
\label{th:chain}
 For any graph $G$,
 \[\omega(G) \le \chi_v(G) \le \chi_f(G) \le \chi(G).\]
\end{theorem}

\subsection{The Ordinary Vertex Coloring}\label{sec:proper}

Consider a proper $d$-coloring $c:V\rightarrow\{1,\ldots,d\}$ of the host graph $G$.
Fix a number $t\in \{0,\ldots,d\}$.
Our idea is to define $V_1$ as the union of $t$ color classes, and $V_2$ as the remaining vertices.
Clearly, for some choices of the colors the set $\la(T_G)$ for a solution $T_G$ can be large, and then by Lemma~\ref{lem:polynomial} we need to set the parameter $r$ high, which makes the running time large.
However, if we try {\em all} the possible choices of $t$ colors, in at least one of them the set $\la(T_G)$ will be small enough.
This is stated in the following lemma.

\begin{lemma}
\label{lem:color-bound}
 Let $c$ be a given $d$-coloring of $G$.
 Let $T_G$ be a fixed $(k,l)$-tree in $G$.
 There is a choice of $t$ color classes $c_1, \ldots, c_t$ such that if we put $V_1=\bigcup_{i=1}^tc^{-1}(i)$, then
 \[la(T_G)\leq \left( 1 - \frac{x(d-x)}{d(d-1)}\right) k + \left(1-\frac{x}{d}\right)l,\]
 where $x=\lfloor\frac{d+(l/k)(d-1)}{2}\rceil$. In particular, $\displaystyle{|\la(T_G)| \le \left(1 - \frac{\floor{\frac{d}2} \ceil{\frac{d}2}}{d(d-1)}\right)k + \frac{l}2}$.
\end{lemma}

\begin{proof}
 For $i=1,\ldots,d$, let $k_i$ denote the number of nodes of $T_G$ colored by $i$.
 Similarly, let $l_i$ denote the number of leaves of $T_G$ colored by $i$.
 Finally, for $i,j = 1,\ldots,d$, let $k_{i,j}$ denote the number of edges in $T_G$ with one endpoint colored with $i$ and the other colored with $j$.
 Note that

 \begin{eqnarray}
  \sum_{i=1}^d k_i = k\\
  \sum_{i=1}^d l_i = l\\
  \sum_{1\le i < j\le d} k_{i,j} = k-1.
 \end{eqnarray}

 Fix $t=0,\ldots,d$.
 It follows that the average size of $|\la(T_G)|$, over all possible choices of the set $S$ of $t$ colors, equals

 \[
 \begin{split}
    & \frac{1}{{d \choose t}} \sum_{S\in {[d]\choose t}} \left( \sum_{i\in S}k_i + \sum_{i\not\in S}l_i + \sum_{\{i,j\}\cap S = \emptyset}k_{i,j} \right) = \\
    & \frac{{d-1 \choose t - 1} k + {d-1 \choose t}l + {d-2 \choose t} (k-1)}{{d \choose t}} \le \\
    & \left( 1 - \frac{t(d-t)}{d(d-1)}\right) k + \frac{d-t}{d}l.
   \end{split}
\]
For $t=x$ and $t=\ceil{d/2}$, we get the claimed bounds.\footnote{We chose $t=x$ in order to minimize the expression $\left( 1 - \frac{t(d-t)}{d(d-1)}\right) k + \frac{d-t}{d}l$. The choise $t=\ceil{d/2}$ gives a simpler bound, but, if $l$ is large,  it is greater than the one with $x$.}
\end{proof}

By combining Lemma~\ref{lem:single-eval-algorithm} with Lemma~\ref{lem:color-bound} we get the following theorem.

\begin{theorem}
\label{th:algorithm-color}
 Assume we are given a proper $d$-coloring of the host graph, for some fixed $d$.
 Then, there is a Monte-Carlo polynomial space algorithm for finding a $(k,l)$-tree running in time $\displaystyle{O^*\left(2^{\left( 1 - \frac{x(d-x)}{d(d-1)}\right) k + \left(1-\frac{x}{d}\right)l}\right)}$, where $x=\lfloor\frac{d+(l/k)(d-1)}{2}\rceil$. In particular, the running time can be bounded by $\displaystyle{O^*\left(2^{\left(1 - \frac{\floor{\frac{d}2} \ceil{\frac{d}2}}{d(d-1)}\right)k + \frac{l}2}\right)}$. The algorithm never reports a false positive, and reports false negatives with constant probability.
\end{theorem}

\subsection{Fractional Vertex Coloring}\label{sec:fractional}

There is a considerable gap between the running time of the $(k,l)$-tree algorithm for 2-colorable graphs and 3-colorable graphs: $O^*(1.415^{k+l})$ vs $O^*(1.588^k1.415^l)$.
It is natural to ask whether for graphs with chromatic number 3, which are in some sense {\em almost 2-colorable}, one can do better than just run the algorithm for 3-colorable graphs.
To address this and related questions in this section, we consider a relaxed concept of graph coloring, called the fractional chromatic number.
We say that a graph $G$ has an $(a:b)$-coloring if, to each vertex of $G$, one can assign a $b$-element subset of $\{1, 2, 3, \dots , a\}$ in such a way that adjacent vertices are assigned disjoint subsets.
The {\em fractional chromatic number} is defined as $\chi_f (G) = \inf \{\frac{a}{b}\ |\ \text{$G$ can be $(a:b)$-colored}\}$.
To avoid confusion let us clarify that for $\frac{a}{b}=\frac{c}{d}$, the fact that $G$ is $(a:b)$-colorable does not imply that it is $(c:d)$-colorable (for an example consider Kneser graph $K_{6:2}$ which is $(6:2)$-colorable but not 3-colorable, see e.g. the monograph of Scheinerman and Ullman~\cite{ScheinermanUllman}). In particular, not every $(a:b)$-colorable graph is $\ceil{\frac{a}{b}}$-colorable.
Indeed, the gap between the fractional chromatic number and the chromatic number can be arbitrarily large.

Now we show an analog of Lemma~\ref{lem:color-bound}.

\begin{lemma}
\label{lem:fract-color-bound}
 Let $c:V(G)\rightarrow {[a] \choose b}$ be a given $(a:b)$-coloring of $G$.
 Let $T_G$ be a fixed $(k,l)$-tree in $G$.
 There is a set of $t$ colors $S\in{[a]\choose t}$ such that if we put $V_1=\{v\in V(G)\ |\ c(v)\cap S = \emptyset\}$ and $V_2=V(G)\setminus V_1$, then
 \[|\la(T_G)| \le \left(1 - \frac{{{a-b}\choose t}-{{a-2b}\choose t}}{{a \choose t}}\right)k + \left(1 - \frac{{{a-b}\choose t}}{{a \choose t}}\right)l.\]
 In particular, there is a color $r\in[a]$ such that if we put $V_1=\{v\in V(G)\ |\ r\not\in c(v)\}$ and $V_2=V(G)\setminus V_1$, then
 \[|\la(T_G)| \le \left(1 - \frac{b}{a}\right)k + \frac{b}{a}l.\]
\end{lemma}

\begin{proof}
 For a given set of $t$ colors $S$, let $V_1^S=\{v\in V(G)\ |\ c(v)\cap S = \emptyset\}$ and $V_2^S=V\setminus V_1^S$.
 For two indices $i,j\in\{1,2\}$ define the set
 \[E^S_{ij} = \{uv \in E(T_G)\ |\ \text{$u$ is the parent of $v$, $u\in V_i^S$ and $v\in V_j^S$}\}.\]
 Also, let $I^S_{i}$ be the set of internal vertices of $T_G$ in $V_i$.
 Similarly, let $L^S_{i}$ be the set of leaves of $T_G$ in $V_i$.
 Then, by the definition of $\la(T_G)$,
\[
\begin{split}
|\la(T_G)| =\ & l + |I^S_1| + |E^S_{22}| = \\
                & l + |I^S_1| + |V_2^S| - |E^S_{12}| - [\treeroot(T_G)\in V_2^S] \le \\
                & l + k -|L^S_1| - |E^S_{12}|.
\end{split}
\]
 For $I\in {[a]\choose b}$, let $k_I$ denote the number of nodes of $T_G$ colored by set $I$.
 Similarly, let $l_I$ denote the number of leaves of $T_G$ colored by $I$.
 Finally, for $I,J\in{[a]\choose b}$, let $k_{I,J}$ denote the number of edges in $T_G$ with the parent colored with $I$ and the child colored with $J$.
 Note that

 \begin{eqnarray}
  \sum_{I\in{[a]\choose b}} k_I = k\\
  \sum_{I\in{[a]\choose b}} l_I = l\\
  \sum_{I,J\in{[a]\choose b}} k_{I,J} = k-1.
 \end{eqnarray}

 It follows that the average size of $|\la(T_G)|$, over all possible choices of the set $S$ of $t$ colors, equals

 \[
 \begin{split}
    & \frac{1}{{a \choose t}} \sum_{S\in {[a]\choose t}} \left(k - |E^S_{12}| + l - |L^S_1|\right) = \\
    & \frac{1}{{a \choose t}} \sum_{S\in {[a]\choose t}} \Big(k - \sum_{\substack{I,J\in{[a]\choose b}\\I\cap S = \emptyset\\ J\cap S \ne\emptyset}}k_{I,J} + l - \sum_{\substack{I\in{[a]\choose b}\\I\cap S = \emptyset}}l_I \Big) = \\
    & \frac{1}{{a \choose t}} \sum_{S\in {[a]\choose t}} \Big(k - \sum_{\substack{I,J\in{[a]\choose b}\\I\cap S = \emptyset}}k_{I,J} + \sum_{\substack{I,J\in{[a]\choose b}\\I\cap S = J\cap S=\emptyset}}k_{I,J} + l - \sum_{\substack{I\in{[a]\choose b}\\I\cap S = \emptyset}}l_I \Big) = \\
    & \left(1 - \frac{{{a-b}\choose t}-{{a-2b}\choose t}}{{a \choose t}}\right)k + \left(1 - \frac{{{a-b}\choose t}}{{a \choose t}}\right)l.
   \end{split}
\]
This implies the first part of the claim. The second part follows after putting $t=1$.
\end{proof}

By combining Lemma~\ref{lem:single-eval-algorithm} with Lemma~\ref{lem:fract-color-bound} we get the following theorem.

\begin{theorem}
\label{th:algorithm-fract-color}
 Assume we are given a proper $(a:b)$-coloring of the host graph, for some fixed $a,b$.
 Then, for any $t=1,\ldots,b$ there is a Monte-Carlo polynomial space algorithm for finding a $(k,l)$-tree running in time
 \[O^*\Bigg(2^{\Big(1 - \frac{{{a-b}\choose t}-{{a-2b}\choose t}}{{a \choose t}}\Big)k + \Big(1 - \frac{{{a-b}\choose t}}{{a \choose t}}\Big)l}\Bigg)\]
 which never reports a false positive, and reports false negatives with constant probability.
\end{theorem}

Note that $t=1$ gives time of $O^*(2^{\left(1 - \frac{b}{a}\right)k + \frac{b}{a}l})$.
For fixed values of $a$ and $b$, it is easy to find optimal value of $t$ in Theorem~\ref{th:algorithm-fract-color}.

The following equivalence is well-known (see e.g.~\cite{kral}).

\begin{proposition}[folklore]
\label{prop:fractional-sampling}
Graph $G$ has fractional chromatic number $\chi_f (G) \le r$ iff there exists a probability distribution $\pi$ on the independent sets of $G$
such that for each vertex $v$, the probability that $v$ is contained in a random independent set (with respect to $\pi$) is at least $1/r$.
\end{proposition}

Proposition~\ref{prop:fractional-sampling} is interesting because an upper bound $B$ on the fractional chromatic number can be proven by providing a polynomial-time randomized algorithm which outputs an independent set $I$ such that, for every vertex $v$, the probability that $v\in I$ is at least $1/B$. This strategy was used e.g.\ by Ferguson et al.~\cite{kral}.
It is not hard to show that having such an algorithm one can build in randomized polynomial time an $(a:b)$-coloring of $G$ such that $\frac{a}{b} = B+\epsilon$ for arbitrary fixed $\epsilon>0$.
However, for our applications in $(k,l)$-tree this is not needed as observed in the following lemma (the main observation behind this lemma traces back to Bj\"orklund\cite{bjorklund-hamilton}, Theorem 2).

\begin{theorem}
\label{thm:fractional-sampling}
 Assume there is a polynomial-time randomized algorithm $A$ which outputs an independent set $I$ such that, for every vertex $v$ of $G$, the probability that $v\in I$ is at least $p$, for some $p\in(0,1)$.
 Then, there is a Monte-Carlo polynomial space algorithm for finding a $(k,l)$-tree running in time
 \[O^*(2^{\left(1 - p\right)k + pl})\]
 which never reports a false positive, and reports false negatives with constant probability.
\end{theorem}

\begin{proof}(Sketch.)
Put $r=\ceil{\left(1 - p\right)k + pl}+1$.
Let $T_G$ be a $(k,l)$-tree in $G$.
Run the algorithm $A$ and let $I$ be the resulting independent set.
Put $V_1=V\setminus I$ and $V_2=I$.
It is easy to see that $\Expect{|\la(T_G)|} = |V(T_G)\cap V_1| + |L(T_G)\cap V_2| = (1-p)k + pl \le r-1$.
By Markov's inequality, $\Pr{|\la(T_G)| \le r} = \Omega(\frac{1}{r}) = \Omega(\frac{1}{k})$.
We run the algorithm from Lemma~\ref{lem:single-eval-algorithm} and it finds a $(k,l)$-tree with probability $\Omega(\frac{1}{k})$.
By repeating the whole procedure $O(k)$ times, we can increase the success probability to a positive constant.
\end{proof}

%

\subsection{Vector Coloring}\label{sec:vector}

A \emph{vector coloring} of a graph $G$ on $n$ vertices of value $r>1$ is an assignment of unit length vectors $\mathbf{v}_u$ in $\mathbb{R}^n$ for each vertex $u$ so that for every edge $uv$, $\mathbf{v}_u \cdot \mathbf{v}_v \leq -1/(r-1)$.  The \emph{vector chromatic number} $\chi_v(G)$ is the smallest real $r$ admitting a vector coloring.

Once we have a vector coloring, we can use a random hyperplane through the origin with normal $\mathbf{h}$ chosen uniformly over the surface of the unit hypersphere to divide the vertices in $V_1$ and $V_2$: For vertex $u$, we put $u$ in $V_1$ if $\mathbf{v}_u \cdot \mathbf{h} \geq 0$; otherwise, we put it in $V_2$.

By the linearity of expectation, we can upper bound the expected number of labels as the expected number of vertices put in $V_1$, which is half of them, plus an upper bound of the expected number of edges in the tree that have both endpoints in $V_2$, which depend on the angle between these two vertices' vectors. The vector coloring value gives us a lower bound on the size of this angle. Paraphrasing Lemma 3.2 in Goemans and Williamson~\cite{GW}, we have
\begin{lemma}
For a random hyperplane normal $\mathbf{h}$ and two vectors $\mathbf{v}_u$ and $\mathbf{v}_v$, we have
\[\operatorname{Pr}(\operatorname{sign}(\mathbf{v}_u\cdot \mathbf{h})\neq \operatorname{sign}(\mathbf{v}_v\cdot \mathbf{h})) = \frac{1}{\pi}\operatorname{arccos}(\mathbf{v}_u\cdot \mathbf{v}_v).
\]
\end{lemma}
Hence we have
\[
\mathbf{E}(|\la(T_G)|)\leq \frac{k+l}{2}+\left(1-\frac{\operatorname{arccos}(-1/(r-1))}{\pi}\right)\frac{k-1}{2}.
\]
By Markov's inequality and by using that $|\la(T_G)|$ is an integer quantity, we see that such a random hyperplane gives us at most the expected number of labels with probability at least $1/(\mathbf{E}(|\la(T_G)|)+1)$.
If we pick $\mathbf{E}(|\la(T_G)|)+1)$ hyperplanes independently of each other, with probability at least $1-e^{-1}$ one of them will use at most $\mathbf{E}(|\la(T_G)|)$ labels. Since $\mathbf{E}(|\la(T_G)|)$ is a polynomial in $k$, we can run Lemma~\ref{lem:single-eval-algorithm} for each hyperplane to obtain the following algorithm guarantee.
\begin{theorem}
\label{th:algorithm-colorv}
 Assume we are given a vector coloring of the host graph of some value $r$.
 Then there is a Monte-Carlo polynomial space algorithm for finding a $(k,l)$-tree running in time
 \[O^*\left(2^{\left( \frac{k+l}{2}+\left(1-\frac{\operatorname{arccos}(-1/(r-1))}{\pi}\right)\frac{k-1}{2}\right)} \right)\]
 which never reports a false positive, and reports false negatives with constant probability.
\end{theorem}

\section{Corollaries}

\subsection{Solving {\sc $(k,l)$-Tree} and {\sc $k$-IST} in Bounded Degree Graphs}\label{sec:bounddeg}

For a graph of bounded degree $\Delta$ that is neither complete nor an odd cycle, it is possible to construct a proper $\Delta$-coloring in linear time (e.g., by Lov\'asz's proof of Brooks' theorem \cite{deltacol}). Note that, given a graph that is either complete or an odd cycle, {\sc $(k,l)$-Tree} can clearly  be solved in polynomial time. In conjunction with Theorem \ref{th:algorithm-color}, this yields the following result.

\begin{corollary}
\label{cor:(k,l)-tree-bouded-degree}
There is a Monte-Carlo polynomial space algorithm for {\sc $(k,l)$-Tree} in graphs of bounded degree $\Delta$ which runs in time
$\displaystyle{O^*\left(2^{\left( 1 - \frac{x(\Delta-x)}{\Delta(\Delta-1)}\right) k + \left(1-\frac{x}{\Delta}\right)l}\right)}$, where $x=\lfloor\frac{\Delta+(l/k)(\Delta-1)}{2}\rceil$. In particular, the running time can be bounded by $\displaystyle{O^*\left(2^{\left(1 - \frac{\floor{\frac{\Delta}2} \ceil{\frac{\Delta}2}}{\Delta(\Delta-1)}\right)k + \frac{l}2}\right)}$.
With $l=2$, the same result holds for \probkPath, and with $k=n$, for \probHamilton.
\end{corollary}
In particular, if $\Delta < 13$ we get \probkPath and \probHamilton algorithms in $O^*(1.655^k)$ and $O^*(1.655^n)$ time, respectively, that are faster than the algorithms from~\cite{BHKK-narrow-sieves, bjorklund-hamilton}.

As shown in \cite{zehavi-ipec13}, a graph $G$ of bounded degree $\Delta$ has a spanning tree with at least $k$ internal vertices {\em iff} it contains a tree with exactly $k$ internal vertices and at most $k - \frac{k-2}{\Delta-1} = \frac{\Delta-2}{\Delta-1}k+O(1)$ leaves.\footnote{For graphs with a {\em small} bounded degree, this bound is better than the one of Strategy A in Section \ref{sec:kIST} (i.e., $l\leq 0.8628 k$). Indeed, for $\Delta\leq 8$, this bound implies that $l\leq 0.8572 k$.}
Therefore, by the above corollary,\footnote{Note that now, in the context of $k$-IST, $k$ is the number of internal vertices, where in Corollary \ref{cor:(k,l)-tree-bouded-degree}, it is the total number of vertices.} we can solve {\sc $k$-IST} in time $\displaystyle{O^*\left(2^{\left( 1 - \frac{x(\Delta-x)}{\Delta(\Delta-1)}\right)\left(1+ \frac{\Delta-2}{\Delta-1}\right)k + \left(1-\frac{x}{\Delta}\right)\left(\frac{\Delta-2}{\Delta-1}\right)k}\right)}$, where $x=\lfloor\frac{3\Delta^2-6\Delta+2}{2(2\Delta-3)}\rceil$.
That is, we can solve {\sc $k$-IST} in time $\displaystyle{O^*\left(2^{\left[(\frac{3\Delta-5}{\Delta-1})-(\frac{x(3\Delta^2+3x+2-6\Delta-2x\Delta)}{\Delta(\Delta-1)^2})\right]k}\right)}$.
Moreover, the worse running time is obtained when $k+l=n$ {\em and} $l=\frac{\Delta-2}{\Delta-1}k$. Therefore, the worse running time is obtained when $k=\frac{\Delta-1}{2\Delta-3}n$, which makes our algorithm run in time $\displaystyle{O^*\left(2^{\left[(\frac{3\Delta-5}{\Delta-1})-(\frac{x(3\Delta^2+3x+2-6\Delta-2x\Delta)}{\Delta(\Delta-1)^2})\right]\left[\frac{\Delta-1}{2\Delta-3}\right]n}\right)}$, which is equal to $\displaystyle{O^*\left(2^{\left[(\frac{3\Delta-5}{2\Delta-3})-(\frac{x(3\Delta^2+3x+2-6\Delta-2x\Delta)}{\Delta(\Delta-1)(2\Delta-3)})\right]n}\right)}$. Thus, we also have the following result.

\begin{corollary}\label{cor:kISTbounded}
There is a Monte-Carlo polynomial space algorithm for {\sc $k$-IST} in graphs of bounded degree $\Delta$ which runs in time
\[\displaystyle{O^*(\min\left\{2^{\left[(\frac{3\Delta-5}{\Delta-1})-(\frac{x(3\Delta^2+3x+2-6\Delta-2x\Delta)}{\Delta(\Delta-1)^2})\right]k},2^{\left[(\frac{3\Delta-5}{2\Delta-3})-(\frac{x(3\Delta^2+3x+2-6\Delta-2x\Delta)}{\Delta(\Delta-1)(2\Delta-3)})\right]n}\right\})},\]
where $x=\lfloor\frac{3\Delta^2-6\Delta+2}{2(2\Delta-3)}\rceil$.
\end{corollary}

\begin{table}[center]
\centering
\begin{tabular}{|l|c|c|c|c|c|}
	\hline
	Reference                         & Variant    & $\Delta=3$       & $\Delta=4$       & $\Delta=5$       & $\Delta=6$       \\\hline\hline		
	Zehavi~\cite{zehavi-ipec13}       & Directed   & $O^*(2.51985^k)$ & $O^*(2.99662^k)$ & $O^*(3.24901^k)$ & $O^*(3.40267^k)$ \\\hline
	Corollary \ref{cor:kISTbounded}   & Undirected & $O^*(2.24493^k)$ & $O^*(2.66968^k)$ & $O^*(2.87787^k)$ & $O^*(3.00355^k)$ \\\hline\hline		
	Raible et al.~\cite{kISPbounddeg} & Undirected & $O^*(1.96799^n)$ & $O^*(1.98735^n)$ & $O^*(1.99376^n)$ & $O^*(1.99777^n)$ \\\hline	
	Corollary \ref{cor:kISTbounded}   & Undirected & $O^*(1.71449^n)$ & $O^*(1.80251^n)$ & $O^*(1.82948^n)$ & $O^*(1.84227^n)$ \\\hline		
\end{tabular}\smallskip
\caption{Some concrete figures for running times of algorithms for {\sc $k$-IST} in bounded degree graphs.}
\label{tab:knownresults3}
\end{table}

\subsection{Subcubic Graphs of Large Girth}
\label{sec:subcubic-large girth}

Fractional chromatic number has been intensively studied for graphs of maximum degree 3 (aka subcubic graphs) of large girth.
Perhaps the most important result in this field is due to Dvo\v{r}\'{a}k, Sereni and Volec~\cite{dvorak}, stating that every subcubic triangle-free graph has fractional chromatic number at most $14/5$.
Unfortunately, their proof does not provide a polynomial-time algorithm for finding such a coloring.
By Lemma~\ref{lem:fract-color-bound}, such an algorithm would imply an algorithm for $(k,l)$-tree running in time
$O^*(2^{\frac{9}{14}k + \frac{5}{14}l})=O^*(1.562^k1.281^l)$.
However, we may use an earlier result of Liu~\cite{liu}, which gives a bound of $\frac{43}{15}$.
Liu's proof is inductive and can be turned to a polynomial-time algorithm which finds a $(516:180)$-coloring~\cite{liu-personal}.
By the second claim of Lemma~\ref{lem:fract-color-bound}, we get the following corollary.

\begin{corollary}
  There is a $O^*(1.571^k 1.274^{l})$-time Monte-Carlo polynomial space algorithm for finding a $(k,l)$-tree in a subcubic triangle-free graph.
\end{corollary}

Hatami and Zhu~\cite{hatami} analyzed the following simple algorithm for sampling an independent set $I$: choose a random bijection $\pi:[n]\rightarrow V(G)$, and for every $i=1,\ldots,n$, include $\pi(i)$ in the independent set $I$ if none of its neighbors is included in $I$.
They show lower bounds $p_k$ for the probability that any vertex $v$ is included in $I$, assuming that $G$ is subcubic and has girth at least $2k+1$.
In particular, they proved that $p_3=1/2.78571\ge 0.3589$, $p_5=1/2.67215\ge 0.3742$ and $p_7=1/2.66681 \ge 0.3749$.
These bounds, together with Lemma~\ref{thm:fractional-sampling} imply Monte-Carlo algorithms for finding a $(k,l)$-tree in a subcubic graphs of girth at least 7, 11 and 15 running in times $O^*(1.5595^k1.283^l)$, $O^*(1.5431^k1.2962^l)$ and $O^*(1.5423^k1.2969^l)$, respectively.

\subsection{{\sc $k$-Path} in Undirected Subcubic Graphs}
\label{sec:triangle-elim}

In this section we reduce the {\sc $k$-Path} problem in subcubic graphs to vertex-weighted {\sc $k$-Path} in triangle-free subcubic graphs.
This is inspired by a similar technique used by Feder, Motwani and Subi~\cite{FederMS02} and Eppstein~\cite{Eppstein07} for finding long cycles, where triangles are removed at a cost of introducing weights at edges.
By plugging in results from Section~\ref{sec:subcubic-large girth}, we get an $O(1.571^k)$-time algorithm.
This is faster than the algorithm running in time $O(2^{2/3k})=O(1.588^k)$ that follows from Corollary~\ref{cor:(k,l)-tree-bouded-degree}.
The improvement we get is admittedly modest, but at least it shows that the natural barrier of $O(2^{2/3k})$ can be broken.

Let $G$ be a graph and let $w:V\rightarrow\mathbb{N}$ be a weight function.
Let vertices $a$, $b$ and $c$ form a triangle in $G$.
By $G/abc$ we denote the graph obtained by contracting vertices $a$, $b$ and $c$ into a single vertex $t$, without introducing multiple edges or loops (this is sometimes called a $\Delta-Y$ transform).
Let $w/abc:V(G/abc)\rightarrow\mathbb{N}$ be a weight function defined by
\[
w/abc(x) = \begin{cases}
w(x) & \mbox{if } x\ne t, \\
w(a)+w(b)+w(c) & \mbox{otherwise}.\\
\end{cases}
\]

\begin{lemma}
\label{lem:triangle-elim}
 Let $(G,w)$ be an undirected vertex-weighted subcubic graph, and let vertices $a$, $b$ and $c$ form a triangle in $G$.
 If $(G,w)$ contains a path of weight $k$ and length $d$, then $(G/abc,w/abc)$ contains a path of weight at least $k$ and length at most $d$.
 Moreover, if $(G/abc,w/abc)$ contains a path of weight $k$, then $(G,w)$ contains a path of weight $k$.
\end{lemma}

\begin{proof}
 Let $P$ be a path of largest weight in $(G,w)$ among paths of length at most $d$. 
 If $\{a,b,c\}\cap V(P) = \emptyset$ then $P$ is a path of weight $w(P)$ in $(G/abc,w/abc)$.
 Hence assume $\{a,b,c\}\cap V(P) \ne \emptyset$.
 We claim that there is a path $P_1$ of length at most $d$ in $(G,w)$ such that vertices $\{a,b,c\}\cap V(P)$ appear consecutively in $P$.
 Indeed, if $\{a,b,c\}\cap V(P)$ do not appear consecutively in $P$ then for two of them, say for $a$ and $b$, for some integer $p\ge 1$, path $P$ contains a subpath $a v_1 \dots v_p b$, such that $c\not\in\{v_1,\ldots,v_p\}$.
 Since $G$ is subcubic, one of the vertices $a$, $b$ is an endpoint of $P$; assume w.l.o.g.\ that $a$ is an endpoint of $P$.
 Then we set $P_1 = P - \{bv_p\} \cup \{ab\}$. Note that $P_1$ has the same length as $P$ and if $c\in V(P_1)$, then $c$ is a neighbor of $b$ in $P_1$.
 Hence $P_1$ is the claimed path.
 Let $t$ be the vertex of $G/abc$ obtained by contracting triangle $abc$.
 After replacing $cba$ in $P_1$ by $t$, we get a path of weight at least $w(P)$ in $(G/abc,w/abc)$.

 Consider a path $P'$ in $(G/abc,w/abc)$.
 If $t\not\in V(P')$, then $P'$ is a path of weight $w/abc(P')$ in $(G,w)$.
 Otherwise it is easy to see that we can obtain a path of weight $w/abc(P')$ in $(G,w)$ by replacing $t$ in $P'$ by one of the paths $abc$, $bca$ or $cab$.
\end{proof}

\begin{corollary}
\label{cor:triangle-elim}
 For a given an undirected subcubic graph $G$ one can build in $O(n)$ time a weighted subcubic triangle-free graph $(G',w)$ such that
 $G$ contains a path of length $k$ iff $G'$ contains a path of weight at least $k$ and length at most $k$.
\end{corollary}

\begin{proof}
Assign weight 1 to every vertex and apply Lemma~\ref{lem:triangle-elim} as long as the graph contains a triangle.
\end{proof}

Now we observe that Lemma~\ref{lem:single-eval-algorithm} generalizes to vertex weighted graphs, provided that the weights are integral, positive and not too big, similarly as shown in~\cite{bjorklund-hamilton}.

\begin{lemma}
\label{lem:single-eval-algorithm-weighted}
Let $W$ be a positive integer.
Let $G$ be an undirected graph with a vertex-weight function $w:V(G)\rightarrow [W]$.
Let $V(G)=V_1 \cup V_2$ be a fixed bipartition of the vertex set of the host graph $G$.
There is an algorithm running in $O^*(2^rW+W\log W)$ time and polynomial space such that
\begin{itemize}
 \item If $G$ does not contain a $(k,l)$-tree of weight at least $W$, then the algorithm always answers NO,
 \item If $G$ contains a $(k,l)$-tree $T_G$ of weight at least $W$ such that $|\la(T_G)|\le r$, then the algorithm answers YES with probability at least $\frac{1}2$.
\end{itemize}
\end{lemma}

\begin{proof}
 (Sketch.)
 We add an additional indeterminate $\eta$ in such a way that a vertex $x$ of weight $w(x)$ contributes $\eta^{w(x)}$ to the monomial.
 More precisely, for a branching walk $B$ and labelling $\ell$, we redefine the corresponding monomial as follows.
\[\mon(B,\ell) =z_{h(1)}\prod_{\substack{\{u,v\}\in E(T)\\ u<v}} x_{h(u),h(v)} \prod_{q\in \la(B)}y_{h(q),\ell(q)} \prod_{x\in V(T)} \eta^{w(h(x))}\,,\]
 It is straightforward to verify that the proof of Lemma~\ref{lem:canceling} is still valid in this setting and the polynomial evaluation algorithm can be easily adapted. Let us group the terms of $P_{\downarrow r}$ according to the exponent at $\eta$.
\[P_{\downarrow r} = \sum_{i\ge 0} R_i \eta^i.\]
 Our task reduces now to determine whether there is an integer $i^*\ge W$ such that the polynomial $R_{i^*}$ is non-zero.
 This is done in a standard way, by sampling the values of all indeterminates except for $\eta$.
 Then, by Schwartz-Zippel Lemma, $R_{i^*}(\x,\y,\z)\ne 0$ with probability at least $1/2$, provided that we use a field and of size at least $6k$ (and still of characteristic two).
 The value of the largest integer $i$ such that $R_{i}(\x,\y,\z)\ne 0$ can be determined using interpolation in $O(W\log W)$ time after $W$ evaluations of the polynomial, see e.g.~\cite{bjorklund-hamilton} or~\cite{BKK-motif-stacs} for details.
\end{proof}

By combining Corollary~\ref{cor:triangle-elim}, Lemma~\ref{lem:single-eval-algorithm-weighted} (for $l=2$), Lemma~\ref{lem:fract-color-bound} and Liu's bound~\cite{liu} similarly as in Section~\ref{sec:subcubic-large girth}, we get the following corollary.

\begin{corollary}
\label{cor:k-path-subcubic}
  There is a $O^*(1.571^k)$-time Monte-Carlo polynomial space algorithm for finding a $k$-path in an undirected subcubic graph.
\end{corollary}

We note, similarly as we did in Section~\ref{sec:subcubic-large girth}, that if there is a (possibly randomized) polynomial-time algorithm that finds an $a:b$ coloring for $a/b=14/5$, then the above bound improves to $O^*(1.562^k)$.

\subsection{Solving {\sc $(k,l)$-Tree} in Graphs with Small Vector Chromatic Number}\label{sec:vectorap}
One of the most striking powers of semidefinite programming is that it is possible to obtain a vector coloring of value $r=\chi_v(G)+\epsilon$ for any fixed $\epsilon>0$ in polynomial time, see for instance~\cite{Gartner}. Combining such a vector coloring algorithm with Lemma~\ref{th:algorithm-colorv}, we have the following corollary.

\begin{corollary}
\label{cor:(k,l)-tree-vector-coloring}
There is a Monte-Carlo polynomial space algorithm for {\sc $(k,l)$-Tree} that for every fixed constant $\epsilon>0$ runs in time
 \[O^*\left(2^{\left( \frac{k+l}{2}+\left(1-\frac{\operatorname{arccos}(-1/(\chi_v(G)+\epsilon-1))}{\pi}\right)\frac{k-1}{2}\right)} \right).\]
With $l=2$ the same result holds for \probkPath and with $k=n$ for \probHamilton.
\end{corollary}
In particular, if $\chi_v(G)\leq 8$, we get \probkPath and \probHamilton algorithms in $O^*(1.655^k)$ and $O^*(1.655^n)$ time, respectively, that are faster than the previous algorithms \cite{BHKK-narrow-sieves, bjorklund-hamilton}.

Let us observe that thanks to Theorem~\ref{th:chain}, we can apply Corollary~\ref{cor:(k,l)-tree-vector-coloring} whenever we know that a chromatic number or the fractional chromatic number is low, even if no efficient algorithm for finding the coloring is known.
For example, since $\chi_f \le 14/5$ for triangle-free subcubic graphs~\cite{dvorak}, by plugging it into Corollary~\ref{cor:(k,l)-tree-vector-coloring} and using the triangle-removing trick from Section~\ref{sec:triangle-elim}, one gets an $O^*(1.576^k)$-time algorithm for \probkPath in subcubic graphs (which is slightly worse than the bound of Corollary~\ref{cor:k-path-subcubic}, but avoids the use of the complicated Liu's algorithm~\cite{liu}).

\subsection*{Acknowledgments}

We thank Zdenek Dvo\v{r}\'{a}k and Dan Kr\'{a}l' for helpful discussions on fractional chromatic number.

\bibliographystyle{abbrv}
\bibliography{trees}

\begin{thebibliography}{10}

\bibitem{colorcoding}
N.~Alon, R.~Yuster, and U.~Zwick.
\newblock Color coding.
\newblock {\em J. ACM}, 42(4):844--856, 1995.

\bibitem{bjorklund-hamilton}
A.~Bj\"orklund.
\newblock Determinant sums for undirected hamiltonicity.
\newblock {\em SIAM J. on Computing}, 43(1):280--299, 2014.

\bibitem{BHKK-TSP}
A.~Bj{\"{o}}rklund, T.~Husfeldt, P.~Kaski, and M.~Koivisto.
\newblock The travelling salesman problem in bounded degree graphs.
\newblock In {\em Proc. ICALP'08}, pages 198--209, 2008.

\bibitem{BHKK-narrow-sieves}
A.~Bj{\"o}rklund, T.~Husfeldt, P.~Kaski, and M.~Koivisto.
\newblock Narrow sieves for parameterized paths and packings.
\newblock {\em CoRR}, abs/1007.1161, 2010.

\bibitem{BKK-motif-stacs}
A.~Bj{\"o}rklund, P.~Kaski, and L.~Kowalik.
\newblock Probably optimal graph motifs.
\newblock In N.~Portier and T.~Wilke, editors, {\em STACS}, volume~20 of {\em
  LIPIcs}, pages 20--31. Schloss Dagstuhl - Leibniz-Zentrum fuer Informatik,
  2013.

\bibitem{Bodlaender93}
H.~L. Bodlaender.
\newblock On linear time minor tests with depth-first search.
\newblock {\em J. Algorithm}, 14(1):1--23, 1993.

\bibitem{divandcol}
J.~Chen, J.~Kneis, S.~Lu, D.~Molle, S.~Richter, P.~Rossmanith, S.~H. Sze, and
  F.~Zhang.
\newblock Randomized divide-and-conquer: Improved path, matching, and packing
  algorithms.
\newblock {\em SIAM J. on Computing}, 38(6):2526--2547, 2009.

\bibitem{kIOB49k}
N.~Cohen, F.~V. Fomin, G.~Gutin, E.~J. Kim, S.~Saurabh, and A.~Yeo.
\newblock Algorithm for finding $k$-vertex out-trees and its application to
  $k$-internal out-branching problem.
\newblock {\em J. Comput. Syst. Sci.}, 76(7):650--662, 2010.

\bibitem{cut-count}
M.~Cygan, J.~Nederlof, M.~Pilipczuk, M.~Pilipczuk, J.~M.~M. van Rooij, and
  J.~O. Wojtaszczyk.
\newblock Solving connectivity problems parameterized by treewidth in single
  exponential time.
\newblock In {\em Proc. FOCS'11}, pages 150--159, 2011.

\bibitem{thesis11}
J.~Daligault.
\newblock Combinatorial techniques for parameterized algorithms and kernels,
  with applicationsto multicut.
\newblock {\em PhD thesis, Universit´e Montpellier II, Montpellier, H´erault,
  France}, 2011.

\bibitem{outbranchpatent}
A.~Demers and A.~Downing.
\newblock Minimum leaf spanning tree.
\newblock {\em US Patent no. 6,105,018}, August 2013.

\bibitem{DeMilloLipton1978}
R.~A. DeMillo and R.~J. Lipton.
\newblock A probabilistic remark on algebraic program testing.
\newblock {\em Inf. Process. Lett.}, 7:193--195, 1978.

\bibitem{dvorak}
Z.~Dvorak, J.~Sereni, and J.~Volec.
\newblock Subcubic triangle-free graphs have fractional chromatic number at
  most 14/5.
\newblock {\em J. London Math. Society}, 89(3):641--662, 2014.

\bibitem{Eppstein07}
D.~Eppstein.
\newblock The traveling salesman problem for cubic graphs.
\newblock {\em J. Graph Algorithms Appl.}, 11(1):61--81, 2007.

\bibitem{FederMS02}
T.~Feder, R.~Motwani, and C.~S. Subi.
\newblock Approximating the longest cycle problem in sparse graphs.
\newblock {\em {SIAM} J. Comput.}, 31(5):1596--1607, 2002.

\bibitem{kral}
D.~G. Ferguson, T.~Kaiser, and D.~Kr{\'{a}}l'.
\newblock The fractional chromatic number of triangle-free subcubic graphs.
\newblock {\em Eur. J. Comb.}, 35:184--220, 2014.

\bibitem{representative}
F.~Fomin, D.~Lokshtanov, and S.~Saurabh.
\newblock Efficient computation of representative sets with applications in
  parameterized and exact agorithms.
\newblock In {\em SODA}, pages 142--151, 2014.

\bibitem{kISP8k}
F.~V. Fomin, S.~Gaspers, S.~Saurabh, and S.~Thomass$\acute{\mathrm{e}}$.
\newblock A linear vertex kernel for maximum internal spanning tree.
\newblock {\em J. Comput. Syst. Sci.}, 79(1):1--6, 2013.

\bibitem{kIOB16k}
F.~V. Fomin, F.~Grandoni, D.~Lokshtanov, and S.~Saurabh.
\newblock Sharp separation and applications to exact and parameterized
  algorithms.
\newblock {\em Algorithmica}, 63(3):692--706, 2012.

\bibitem{productFam}
F.~V. Fomin, D.~Lokshtanov, F.~Panolan, and S.~Saurabh.
\newblock Representative sets of product families.
\newblock In {\em ESA}, pages 443--454, 2014.

\bibitem{Gartner}
B.~G{\"a}rtner and J.~Mat{\v{o}}usek.
\newblock {\em Approximation algorithms and semidefinite programming}.
\newblock Springer, 2012.

\bibitem{Gebauer08}
H.~Gebauer.
\newblock On the number of hamilton cycles in bounded degree graphs.
\newblock In {\em Proc. ANALCO'08}, pages 241--248, 2008.

\bibitem{GW}
M.~X. Goemans and D.~P. Williamson.
\newblock Improved approximation algorithms for maximum cut and satisfiability
  problems using semidefinite programming.
\newblock {\em Journal of the ACM}, 42(6):1115--1145, 1995.

\bibitem{domset-breaking}
F.~Grandoni.
\newblock A note on the complexity of minimum dominating set.
\newblock {\em J. Discrete Algorithms}, 4(2):209--214, 2006.

\bibitem{kIOB2klogk}
G.~Gutin, I.~Razgon, and E.~J. Kim.
\newblock Minimum leaf out-branching and related problems.
\newblock {\em Theor. Comput. Sci.}, 410(45):4571--4579, 2009.

\bibitem{GvozdenovicL08}
N.~Gvozdenovic and M.~Laurent.
\newblock The operator psi for the chromatic number of a graph.
\newblock {\em {SIAM} Journal on Optimization}, 19(2):572--591, 2008.

\bibitem{hatami}
H.~Hatami and X.~Zhu.
\newblock The fractional chromatic number of graphs of maximum degree at most
  three.
\newblock {\em {SIAM} J. Discrete Math.}, 23(4):1762--1775, 2009.

\bibitem{iwama}
K.~Iwama and T.~Nakashima.
\newblock An improved exact algorithm for cubic graph {TSP}.
\newblock In {\em Proc. COCOON'07}, pages 108--117, 2007.

\bibitem{koutis-icalp08}
I.~Koutis.
\newblock Faster algebraic algorithms for path and packing problems.
\newblock In {\em Proc. ICALP'08}, volume 5125 of {\em LNCS}, pages 575--586,
  2008.

\bibitem{4kist}
W.~Li, J.~Wang, J.~Chen, and Y.~Cao.
\newblock A $2k$-vertex kernel for maximum internal spanning tree.
\newblock {\em CoRR abs/1412.8296}, 2014.

\bibitem{liu}
C.~Liu.
\newblock An upper bound on the fractional chromatic number of triangle-free
  subcubic graphs.
\newblock {\em {SIAM} J. Discrete Math.}, 28(3):1102--1136, 2014.

\bibitem{liu-personal}
C.~Liu.
\newblock personal communication, 2015.

\bibitem{deltacol}
L.~Lov\'asz.
\newblock Three short proofs in graph theory.
\newblock {\em J. Combin. Theory Ser.}, 19:269--271, 1975.

\bibitem{Monien85}
B.~Monien.
\newblock How to find long paths efficiently.
\newblock {\em Annals of Discrete Mathematics}, 25:239--254, 1985.

\bibitem{nederlof-steiner}
J.~Nederlof.
\newblock Fast polynomial-space algorithms using {M{\"o}bius} inversion:
  Improving on {Steiner} tree and related problems.
\newblock In {\em Proc. ICALP'09}, volume 5555 of {\em LNCS}, pages 713--725,
  2009.

\bibitem{NedkIST}
J.~Nederlof.
\newblock Fast polynomial-space algorithms using inclusion-exclusion.
\newblock {\em Algorithmica}, 65(4):868--884, 2013.

\bibitem{kISP24klogk}
E.~Prieto and C.~Sloper.
\newblock Reducing to independent set structure -- the case of $k$-internal
  spanning tree.
\newblock {\em Nord. J. Comput.}, 12(3):308--318, 2005.

\bibitem{kISPbounddeg}
D.~Raible, H.~Fernau, D.~Gaspers, and M.~Liedloff.
\newblock Exact and parameterized algorithms for max internal spanning tree.
\newblock {\em Algorithmica}, 65(1):95--128, 2013.

\bibitem{fvs-breaking}
I.~Razgon.
\newblock Exact computation of maximum induced forest.
\newblock In {\em Proc. SWAT'06}, pages 160--171, 2006.

\bibitem{ScheinermanUllman}
E.~R. Scheinerman and D.~H. Ullman.
\newblock {\em Fractional graph theory}.
\newblock Wiley-Interscience Series in Discrete Mathematics and Optimization.
  John Wiley \& Sons Inc., 1997.

\bibitem{schwartz}
J.~T. Schwartz.
\newblock Fast probabilistic algorithms for verification of polynomial
  identities.
\newblock {\em J. ACM}, 27(4):701--717, 1980.

\bibitem{esarepresentative}
H.~Shachnai and M.~Zehavi.
\newblock Representative families: A unified tradeoff-based approach.
\newblock In {\em ESA}, pages 786--797, 2014.

\bibitem{williams-ipl}
R.~Williams.
\newblock Finding paths of length $k$ in {$O^*(2^k)$} time.
\newblock {\em Inf. Process. Lett.}, 109(6):315--318, 2009.

\bibitem{zehavi-ipec13}
M.~Zehavi.
\newblock Algorithms for k-internal out-branching.
\newblock In {\em Proc. IPEC'13}, volume 8246 of {\em LNCS}, pages 361--373.
  Springer, 2013.

\bibitem{zehavi-arxiv}
M.~Zehavi.
\newblock Solving parameterized problems by mixing color coding-related
  techniques.
\newblock {\em CoRR}, abs/1410.5062, 2014.

\bibitem{zippel}
R.~Zippel.
\newblock Probabilistic algorithms for sparse polynomials.
\newblock In {\em Proc. International Symposium on Symbolic and Algebraic
  Computation}, volume~72 of {\em LNCS}, pages 216--226, 1979.

\end{thebibliography}

\end{document}